\documentclass[reqno]{amsart}
\usepackage[a4paper]{geometry}
\usepackage[utf8]{inputenc} 
\usepackage[english]{babel}
\usepackage{amsmath,amsfonts,amssymb,amscd,bm, amsthm}
\usepackage{latexsym}
\usepackage{graphicx}
\usepackage[margin=1cm]{caption}
\usepackage[font=small]{subcaption}
\usepackage{tikz}
\usepackage{pgfplots}
\pgfplotsset{compat=1.16}

\usepackage{fancyhdr}
\usepackage{color}
\usepackage{listings}
\usepackage{pdfpages}
\usepackage{hyperref}
\usepackage{hhline}
\usepackage{courier}
\usepackage{mathrsfs}
\usepackage{wrapfig}
\usepackage{chngcntr}
\usepackage{multicol}
\usepackage{xcolor,colortbl}
\usepackage{bigfoot}
\usepackage{afterpage}
\usepackage{xurl}
\usepackage{tabularx,colortbl}
\usepackage{algorithm,algpseudocode}
\usepackage[percent]{overpic}

\usepackage[backend=bibtex, style=numeric-comp]{biblatex}
\appto{\bibsetup}{\sloppy}

\newtheorem{theorem}{Theorem}[section]
\newtheorem{lemma}[theorem]{Lemma}
\newtheorem{proposition}[theorem]{Proposition}
\newtheorem{corollary}[theorem]{Corollary}
\newtheorem{conjecture}[theorem]{Conjecture}

\theoremstyle{definition}

\newtheorem{openproblem}[theorem]{Open Problem}

\theoremstyle{remark}
\newtheorem{remark}{Remark}

\newcommand{\dd}{\,\mathrm{d} }

\numberwithin{equation}{section}

\begin{document}

\title{Optimizing the ground state energy of the three-dimensional magnetic Dirichlet Laplacian with constant magnetic field}

\author{Matthias Baur}
\address{Institute of Analysis, Dynamics and Modeling, Department of Mathematics, University of Stuttgart, Pfaffenwaldring 57, 70569 Stuttgart, Germany}
\email{matthias.baur@mathematik.uni-stuttgart.de}

\begin{abstract}
This paper concerns the shape optimization problem of minimizing the ground state energy of the magnetic Dirichlet Laplacian with constant magnetic field among three-dimensional domains of fixed volume. In contrast to the two-dimensional case, a generalized ``magnetic'' Faber-Krahn inequality does not hold and minimizers cannot be balls when the magnetic field is turned on. An analysis of the problem among cylindrical domains reveals geometric constraints for general minimizers. In particular, minimizers must elongate with a certain rate along the direction of the magnetic field as the field strength increases. In addition to the theoretical analysis, we present numerical minimizers which confirm this prediction and give rise to further conjectures.
\end{abstract}

\maketitle

\section{Introduction}

Let $\Omega \subset \mathbb{R}^d$, $d=2$ or $3$, be a bounded, open domain. We consider the magnetic Laplacian with constant magnetic field on $\Omega$ with Dirichlet boundary condition. For domains with sufficiently regular boundary, the corresponding eigenvalue problem is
\begin{align} \label{eq:magn_lap_eigvalproblem}
(-i\nabla + A)^2 u &= \lambda u, \qquad \text{in } \Omega,\\
u&=0, \qquad \text{on } \partial \Omega,
\end{align}
where $A$ is a vector potential such that it induces the magnetic field $B$ via $B=  \operatorname{curl } A$. In two dimensions ($d=2$), it is customary to choose the standard linear vector potential $A(x_1,x_2)= \frac{B}{2}(-x_2 , x_1 )$, where $B\in\mathbb{R}$ is the strength of the magnetic field. In three dimensions ($d=3$), we assume the coordinate system is rotated in such a way that the magnetic field points in the positive direction of the $z$-axis, i.e.\ $B(x) = (0,0,B)^T$ with $B\geq 0$, and we choose the linear vector potential $A(x_1,x_2,x_3)= \frac{B}{2}(-x_2, x_1, 0)^T$. 
 
More generally, the magnetic Dirichlet Laplacian on $\Omega \subset \mathbb{R}^d$ can be defined by Friedrichs extension of the quadratic form
\begin{align}
h_A^\Omega[u] = \int_\Omega |(-i\nabla + A) u|^2 \, dx, \qquad u\in C_0^\infty(\Omega).
\end{align}
The associated self-adjoint operator $H_B^\Omega$ admits an infinite sequence of eigenvalues $\{\lambda_n(\Omega, B) \}_{n\in\mathbb{N}}$ that accumulate at infinity only. Counting multiplicities, we assume the eigenvalues to be sorted in increasing order, so that
\begin{align}
0 < \lambda_1(\Omega,B) \leq \lambda_2(\Omega,B) \leq ... \leq \lambda_n(\Omega, B) \leq ...
\end{align}
Note that when the magnetic field is turned off, i.e.\ when $B=0$, the magnetic Dirichlet Laplacian reduces to the usual Dirichlet Laplacian: $H_0^\Omega = - \Delta_\Omega^D$.

The Dirichlet Laplacian plays a significant role in quantum mechanics. It is used to describe a confined particle, free of any interactions. The magnetic Laplacian introduces an interaction with a homogeneous, external magnetic field. In this model, the particle is assumed to be charged, but spinless. Historically however, before its triumph in quantum mechanics, the Laplacian was introduced by Pierre-Simon Laplace in celestial mechanics and has been employed among others by Joseph Fourier and Lord Rayleigh in the description of vibrating membranes. 

Initially conjectured by Lord Rayleigh, the celebrated Faber-Krahn inequality \cite{Faber1923,Krahn1925} states that among all drums of equal area, the lowest possible fundamental frequency is attained by a disk. In other words and with the notation introduced above, for any bounded, open domain $\Omega \subset \mathbb{R}^2$,
\begin{align}
\lambda_1(\Omega, 0) \geq \lambda_1(D, 0),
\end{align}
where $D$ is a disk with $|D| = |\Omega|$. The same isoperimetric inequality holds for domains $\Omega \subset \mathbb{R}^d$, $d>2$, but with $D$ replaced by a $d$-dimensional ball. In 1996, Erd\H{o}s \cite{Erdoes1996} generalized the Faber-Krahn inequality to the planar magnetic Dirichlet Laplacian. He showed that for any $\Omega \subset \mathbb{R}^2$ and any $B \in \mathbb{R}$,
\begin{align}
\lambda_1(\Omega, B) \geq \lambda_1(D, B),
\end{align}
where $D$ is a disk with $|D| = |\Omega|$. Recent work \cite{Fournais2019, Colbois2023, Colbois2024} has brought a lot of attention to a conjectured ``reverse'' Faber-Krahn inequality for the planar magnetic Neumann Laplacian in two dimensions. The Neumann counterpart of Erd\H{o}s' theorem has been established in \cite{Colbois2024} up to some explicit value of $B$, but remains open for large $B$.

In this work, we want to address another natural question connected with Erd\H{o}s' theorem, namely, whether the Faber-Krahn inequality extends also to the magnetic Dirichlet Laplacian in higher space dimensions, specifically in the physically most relevant case of three-dimensional space. We show that this is not the case and in general, minimizers of the problem 
\begin{align}
\min_{\substack{\Omega \subset \mathbb{R}^3 \text{ open}, \\ |\Omega|=c}} \lambda_1(\Omega, B) \label{eq:min_problem_c}
\end{align}
cannot be balls, see Corollary \ref{cor:no_balls}. But if minimizers for the ground state energy of the magnetic Dirichlet Laplacian exist, what are they and how do they depend on the field strength $B$?

This paper aims to investigate this question analytically and numerically. By analyzing cylindrical domains, we establish that possible minimizers of \eqref{eq:min_problem_c} must elongate along the axis of the magnetic field as the field strength $B$ goes to infinity. Intuitively, this ``spaghettification'' phenomenon is linked to the fact that in $\mathbb{R}^3$, the constant magnetic field introduces a preferred direction and destroys a degree of symmetry. We also quantify the rate of degeneration of the minimizers (Section \ref{sec:cyl}). We supply our analytical results with numerical computations, employing and combining methods from \cite{Antunes2012, Antunes2017, Betcke2005, Baur2025a} (Section \ref{sec:nummeth}). The main novelty of this section is that we use the Method of Particular Solutions to compute eigenvalues for the three-dimensional magnetic Laplacian. In case where the domain is axisymmetric, we lay out further simplifications can lead to increased accuracy of the numerical method. We then describe the minimization procedure applied to solve \eqref{eq:min_problem_c}. Finally, we show plots of the obtained numerical minimizers, discuss the results and formulate several problems that remain open (Section \ref{sec:res}).

\section{Preliminaries} \label{sec:prelim}

\subsection{Basic properties of magnetic eigenvalues}

Before we analyze cylindrical domains, let us recall a few important properties of eigenvalues $\lambda_n(\Omega,B)$ of the magnetic Dirichlet Laplacian in two and three dimensions.

\begin{proposition} \label{prop:simple_est_domain_mono} Let $\Omega \subset \mathbb{R}^d$, $d=2$ or $3$, be a bounded, open domain. Then:
\begin{enumerate}
\item If $d=2$, then the eigenvalues $\lambda_n(\Omega,B)$ are invariant under rigid transformations of $\Omega$. If $d=3$, then the eigenvalues $\lambda_n(\Omega,B)$ are invariant under translations of $\Omega$ and under rotations of $\Omega$ around the $z$-axis.
\item For any $B\geq 0$, 
\begin{align}
\lambda_1(\Omega,B) \geq \max\{ \lambda_1(\Omega,0) , B \}.
\end{align}
\item For any $n \in\mathbb{N}$, the map $B \mapsto \lambda_n(\Omega, B)$ is piecewise real-analytic.
\item If $\Omega'$ is a domain with $\Omega' \subset \Omega$, then 
\begin{align}
\lambda_n(\Omega',B) \geq \lambda_n(\Omega,B)
\end{align}
for any $n \in\mathbb{N}$.
\end{enumerate}
\end{proposition}

\begin{proof}
1. Any rigid transformation of $\Omega$ in the $xy$-plane can be reversed by a change of the coordinate system. This yields the magnetic Dirichlet Laplacian on the original domain $\Omega$ but with a vector potential in a different gauge. The eigenvalues $\lambda_n(\Omega, B)$ are however invariant under gauge transformations of the vector potential. 
2. The lower bound $\lambda_1(\Omega,B) \geq  \lambda_1(\Omega,0) $ is a direct consequence of the well-known diamagnetic inequality, see \cite{Cycon1987, Avron1978, Fournais2010}. The lower bound $\lambda_1(\Omega,B) \geq B$ follows from a commutator estimate \cite{Ekholm2016, Fournais2010, Avron1978}.
3. The operators $\{ H_B^\Omega \}$ form a type (B) self-adjoint holomorphic family and therefore the spectrum of the magnetic Dirichlet Laplacian over $B$ can be described in terms of a countable set of real-analytic eigenvalue curves, see Kato \cite[Chapter VII \S 3 and \S 4]{Kato1995}. After sorting the eigenvalues by size, one obtains that the functions $B\mapsto \lambda_n(\Omega,B)$ are piecewise real-analytic (piecewise, since there may exist crossings of the eigenvalue curves).
4. Follows from the variational principle.
\end{proof}

Additionally, we have the following scaling identity.

\begin{proposition} \label{prop:scaling_identity}
Let $\Omega \subset \mathbb{R}^d$, $d=2$ or $3$, be a bounded, open domain. Then, for any $t>0$,
\begin{align}
\lambda_n\left(t\Omega, \frac{B}{t^2}\right) = \frac{\lambda_n(\Omega,B)}{t^2}.
\end{align}
\end{proposition}
Here, $t\Omega = \{tx \, : \, x \in \Omega \}$ is a homothetic dilation of $\Omega$. The identity follows directly from the definition of the quadratic form $h_A^\Omega$.

The scaling identity has an important consequence for problem \eqref{eq:min_problem_c}. Since 
\begin{align}
\lambda_1(\Omega, B ) = |\Omega|^{- \frac{2}{d}}\lambda_1(|\Omega|^{-\frac{1}{d}}\Omega, B|\Omega|^\frac{2}{d}),
\end{align}
we conclude that 
\begin{align}
\min_{\substack{\Omega \text{ open}, \\ |\Omega|=c}} \lambda_1(\Omega, B)= c^{-\frac{2}{d}} \min_{\substack{\Omega \text{ open}, \\ |\Omega|=1}} \lambda_1(\Omega, Bc^{\frac{2}{d}}) 
\end{align}
This means that if the shape optimization problem \eqref{eq:min_problem_c} is solved for $c=1$ and any $B$, it is automatically solved for arbitrary $c>0$ by appropriate scaling. For this reason, we will focus on the volume-normalized case $c=1$ for the rest of the paper. 

\subsection{Existence of minimizers}

It is of course an important question whether the minimization problem \eqref{eq:min_problem_c} actually admits a minimizer. Direct proofs that construct the minimizing shape through symmetrization like the proofs of Faber and Krahn are bound to fail since the ball is in general not a minimizer. The shape of minimizers must rather depend on the field strength $B$ and cannot be spherically symmetric in general. That does not mean that minimizers have no symmetry at all. We will see in Section \ref{sec:res} that numerical minimizers exhibit a very striking axial symmetry, but we are unable to prove this at this point. 

Even for the non-magnetic Dirichlet Laplacian, general existence results for shape optimization with respect to eigenvalues are a non-trivial matter. An important milestone was set by Buttazzo and Dal Maso \cite{Buttazzo1993} who established the existence of so-called quasi-open optimal shapes for the more general problem of minimizing the $k$-th eigenvalue of the Dirichlet Laplacian, i.e.~for $\lambda_k(\Omega, 0)$, under fixed volume constraint. They worked under the restriction that minimizing shapes are sought within a given bounded design region. To then conclude that minimizers exist also for the problem without bounded design region, one then essentially needs to prove that minimizers stay bounded and do not develop long tails for large design regions. This was shown almost 20 years later by Bucur \cite{Bucur2012} and Mazzoleni and Pratelli \cite{Mazzoleni2013}. 

In Appendix \ref{app:existence_of_minimizers}, we give a proof that establishes existence of a minimizer for the magnetic Dirichlet ground state energy among quasi-open sets within an arbitrary design region, similar to the Buttazzo-Dal Maso result. We expect that the methods of Bucur and Mazzoleni-Pratelli for the whole-space problem could be adapted for the present case of the magnetic Laplacian, but this is well beyond the scope of this article, so we leave it open for future work.

Since existence questions are not the main focus of this article, we work in the following under the working assumption that a (possibly non-unique) open minimizer $\Omega^*(B)$ for \eqref{eq:min_problem_c} exists for any $B\geq 0$. The associated minimal ground state energy is then defined as $\lambda_1^*(B) := \lambda_1(\Omega^*(B), B)$.

\section{Cylinders} \label{sec:cyl}

The goal of this section is to discuss the volume-normalized shape optimization problem \eqref{eq:min_problem_c} among the class of cylindrical domains that are aligned with the $z$-axis. This yields an upper bound for the minimal eigenvalue $\lambda_1^*(B)$ and assertions about the geometry of possible minimizers $\Omega^*(B)$. 

Let $\Omega_{\omega,h} = \omega \times (0,h)$ where $\omega \subset \mathbb{R}^2$ is a planar, bounded, open domain and $h>0$. We are interested in the minimization problem
\begin{align}
\min_{|\Omega_{\omega,h}|=1} \lambda_1(\Omega_{\omega,h}, B). \label{eq:min_problem_generalized_cylinders}
\end{align}
By separation of variables,
\begin{align*}
\lambda_1(\Omega_{\omega,h}, B) = \lambda_{1}(\omega, B)+ \frac{\pi^2}{h^2}.
\end{align*}
The isoperimetric inequality by Erd{\H o}s \cite{Erdoes1996, Ghanta2024} for planar domains implies 
\begin{align*}
\lambda_1(\Omega_{\omega,h}, B) = \lambda_{1}(\omega, B)+ \frac{\pi^2}{h^2} \geq \lambda_{1}(D_R, B)+ \frac{\pi^2}{h^2}
\end{align*}
where $D_R$ is a disk with radius $R>0$ chosen such that $|D_R| = |\omega|$. Note that equality only occurs if $\omega$ is a disk. Thus, it suffices to seek solutions of problem \eqref{eq:min_problem_generalized_cylinders} among circular cylinders.

Let $h,R >0$ and let $C_{R,h} := D_R \times (0, h)$ denote a circular cylinder with base radius $R$ and height $h$. We look for solutions of the minimization problem
\begin{align}
\min_{\substack{|C_{R,h}|=1, \\ h>0}} \lambda_1(C_{R,h}, B). \label{eq:min_problem_cylinders}
\end{align}
Imposing the volume constraint $|C_{R,h}|=1$ gives us the relation $\pi R^2 h = 1$ between the radius $R$ and the height $h$ of the cylinders, making the problem effectively a one parameter problem. 

As noted before, separation of variables yields
\begin{align*}
\lambda_1(C_{R,h}, B) = \lambda_{1}(D_R, B)+ \frac{\pi^2}{h^2}
\end{align*}
where $\lambda_{1}(D_R, B)$ denotes the first eigenvalue of a planar disk of radius $R= (\pi h)^{-1/2}$ at field strength $B$. Using the scaling identity for the disk yields
\begin{align*}
\lambda_1(C_{R,h}, B) = h \lambda_{1}\left(D,\frac{B}{ h}\right)  + \frac{\pi^2}{h^2}
\end{align*}
where $D$ is a disk with $|D|=1$. As $h \mapsto h\lambda_1(D,B/h)$ is real-analytic, strictly increasing with $h$ and satisfies 
\begin{align*}
\lim_{h\to 0} h\lambda_1(D,B/h) = B, \qquad \lim_{h\to  \infty} h\lambda_1(D,B/h) = \infty,
\end{align*}
see the proof of Theorem 2.1 in \cite{Baur2025}, we realize that for any $B \geq 0$, problem \eqref{eq:min_problem_cylinders} admits a solution. The solution shall be the cylinder $C^*(B) = C_{R^*(B),h^*(B)}$ with optimal height $h^*(B)$ and corresponding optimal radius $R^*(B)=(\pi h^*(B))^{-1/2}$. Should the minimizer be non-unique, we pick an arbitrary one. We further define $\lambda_{1,cyl}^*(B):=\lambda_1(C^*(B), B)$ as the optimal ground state energy among cylinders. Figure \ref{fig:h_optcyl} shows $h^*(B)$ and $\lambda_{1,cyl}^*(B)$ obtained by numerical computations.

\begin{figure}
\includegraphics[scale=0.33]{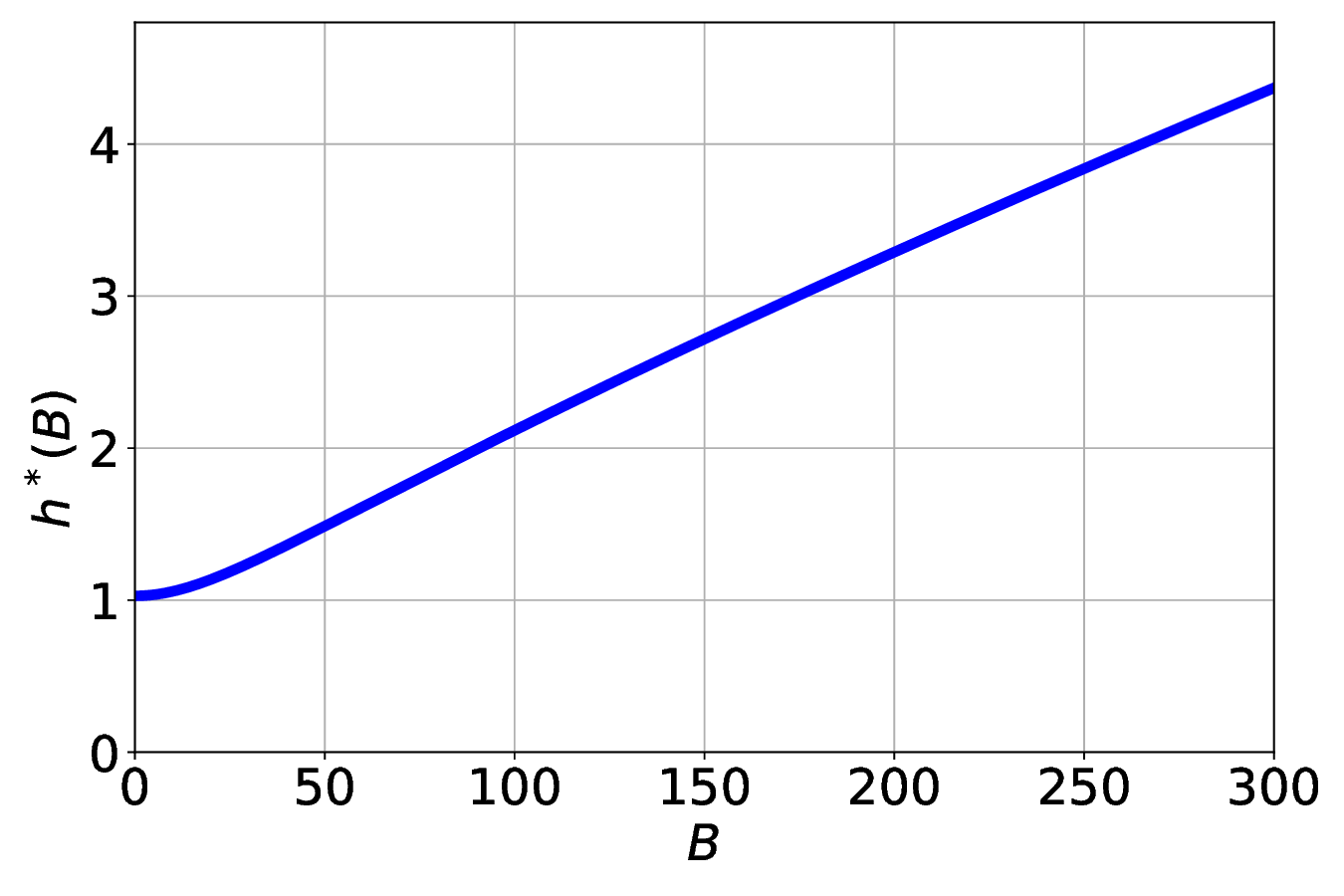}
\includegraphics[scale=0.33]{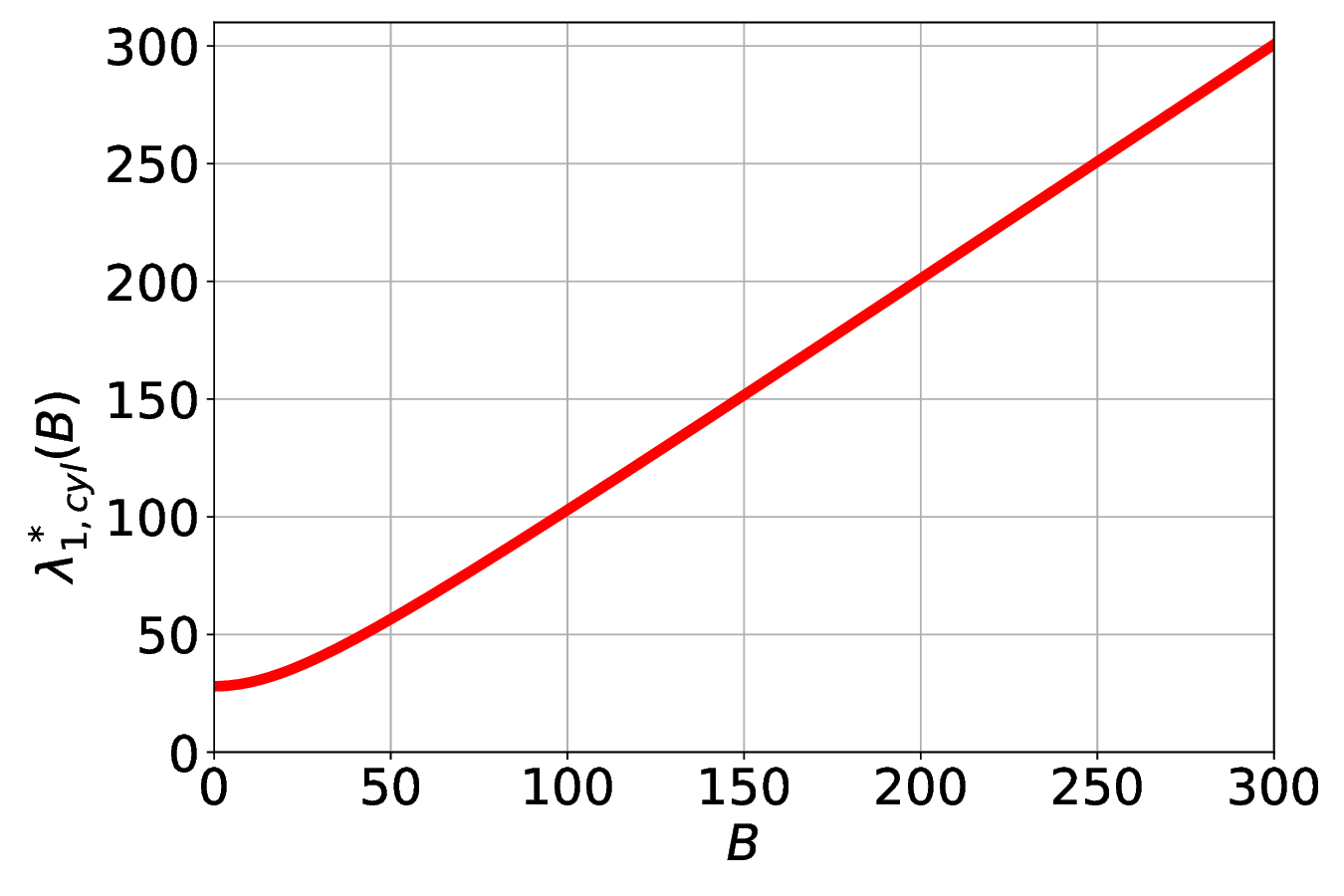}
\caption{$h^*(B)$ and $\lambda_{1,cyl}^*(B)$ as functions of the field strength $B$.} \label{fig:h_optcyl}
\end{figure}

\begin{remark}
If $B=0$, then
\begin{align*}
\lambda_1(C_{R,h}, 0) = \frac{(j_{0,1})^2}{R^2} + \frac{\pi^2}{h^2} = h \pi (j_{0,1})^2 + \frac{\pi^2}{h^2}
\end{align*}
where $j_{0,1}$ denotes the first positive zero of the Bessel function $J_0$. This yields
\begin{align*}
h^*(0) = \left(\frac{2\pi}{(j_{0,1})^2} \right)^{1/3} \approx 1.028 \qquad \text{and} \qquad  \lambda_{1,cyl}^*(0) = \frac{3}{4^{1/3}}  (\pi j_{0,1})^{4/3} \approx 28.016.
\end{align*}
\end{remark}

\medskip
\medskip

The following theorem describes the behavior of the height of the optimal cylinders and the associated ground state energy in the strong field limit.

\begin{theorem} \label{thm:lambda1cyl_asympt}
\begin{align}
h^*(B) = \frac{1}{6\pi} \frac{B}{\log(B)} (1+o(1)) \qquad \text{as } B\to \infty \label{eq:thm_h_star_6pi} 
\end{align}
and
\begin{align}
\lambda_{1,cyl}^*(B) &= B + 36\pi^4 \dfrac{\log(B)^2}{B^2} (1+o(1)) \qquad \text{as } B\to \infty.  \label{eq:thm_h_star_6pi_2} 
\end{align}

\end{theorem}

\begin{proof}
%

First, consider a volume-normalized cylinder $C_{R,h}$ where $h=h(B)$ is a function of the field strength with $B/h(B)\to \infty$ for $B\to \infty$. Then, one can apply the asymptotic expansion for the ground state energy of the disk in the strong field limit, see \cite[Theorem 5.1]{Helffer2017} and \cite[Theorem 2.1]{Baur2025}, which yields
\begin{align}
\lambda_1(C_{R,h(B)}, B) &= h(B) \lambda_{1}\left(D,\frac{B}{ h(B)}\right)  + \frac{\pi^2}{h(B)^2} \\
&=h(B) \left( \frac{B}{h(B)} + \frac{1}{\pi} \left(  \frac{B}{h(B)}\right)^2 \exp\Big(- \frac{B}{2\pi h(B)}\Big)  (1+O((B/h(B))^{-1})) \right)  + \frac{\pi^2}{h(B)^2}\\
&= B + \frac{1}{\pi}  \frac{B^2}{h(B)} \exp\Big(-\frac{B}{2\pi h(B)}\Big)   (1+O((B/h(B))^{-1}))   + \frac{\pi^2}{h(B)^2} \label{eq:lambda_1_expansionallowed}
\end{align}
as $B\to \infty$. Choosing for example $h(B) = B^\alpha$ with $0< \alpha<1$, one finds that 
\begin{align}
\lambda_1(C_{R,h(B)}, B) = B + \frac{\pi^2}{B^{2\alpha}} (1+o(1)) \label{eq:cylinder_Balpha}
\end{align}
as $B\to \infty$. Since by Proposition \ref{prop:simple_est_domain_mono} (2), 
\begin{align}
\lambda_1(C_{R,h}, B) = \lambda_1(D_R,B) + \frac{\pi^2}{h^2} \geq B + \frac{\pi^2}{h^2} \label{eq:easy_cylinder_estimate}
\end{align}
for any $R>0$ and any $h>0$, we conclude that $h^*(B)$ must be unbounded, i.e.\ $h^*(B)\to \infty$ for $B\to \infty$. In fact, equation \eqref{eq:cylinder_Balpha} implies that $h^*(B)$ must grow faster than $B^\alpha$ for any $ \alpha<1$.

Let us now show that $h^*(B)=o(B)$ as $B\to \infty$. Suppose $c>0$ is arbitrary and $h> c B$. Since $h \mapsto h\lambda_1(D,B/h)$ is strictly increasing (this follows for example from the fact that $t \mapsto \lambda_1(tD,B)$ is strictly decreasing with $t$ and the scaling identity Proposition \ref{prop:scaling_identity}), we have
\begin{align}
\lambda_1(C_{R,h}, B) &=  h \lambda_{1}\left(D,\frac{B}{ h}\right)  + \frac{\pi^2}{h^2} > c \lambda_{1}(D,c^{-1}) B. \label{eq:lambda1_CRH_lowerbound}
\end{align}
Now, notice that $c\lambda_1(D,c^{-1})>1$ for any $c>0$, see \cite[Theorem 2.1]{Baur2025}. Comparing the lower bound \eqref{eq:lambda1_CRH_lowerbound} with the asymptotic \eqref{eq:cylinder_Balpha} for the choice $h(B) =B^\alpha$, $ \alpha<1$, we see that cylinders $C_{R,h}$ with $h> c B$ cannot be minimizers of $\lambda_1(\Omega, B)$ when $B$ is large enough. Hence, we must have $h^*(B) \leq c B$ for $B$ large enough. Since $c$ was arbitrary, this shows $h^*(B)=o(B)$ as $B\to \infty$.

We are now ready to prove the asymptotics for $h^*(B)$. Since $h^*(B)=o(B)$ as $B\to \infty$, we can apply the asymptotic expansion \eqref{eq:lambda_1_expansionallowed} to the optimal cylinders $C^*(B)$, i.e.~we have
\begin{align}
\lambda_1(C^*(B), B) &= B + \frac{1}{\pi}  \frac{B^2}{h^*(B)} \exp\Big(-\frac{B}{2\pi h^*(B)}\Big)   (1+O((B/h^*(B))^{-1}))   + \frac{\pi^2}{h^*(B)^2}
\end{align}
as $B\to \infty$. Let $h_C(B)=C B/\log(B)$. Then,
\begin{align}
\lambda_1(C_{R,h_C(B)}, B) &= B + \dfrac{1}{\pi C}   \log(B) B^{1-\frac{1}{2\pi C}}   (1+O(\log(B)^{-1})) + \dfrac{\pi^2}{C^2} \dfrac{\log(B)^2}{B^2}
\end{align}
as $B\to \infty$. If $C=1/(6\pi)$, then 
\begin{align}
\lambda_1(C_{R,h_{1/(6\pi)}(B)}, B) &= B + 6  \frac{ \log(B) }{B^2} (1+O(\log(B)^{-1})) + 36\pi^4 \dfrac{\log(B)^2}{B^2} \label{eq:lambda1_Cequal1_6pi} 
\end{align}
as $B\to \infty$. 

Now, suppose there exists a $C> 1/(6\pi)$ such that $h^*(B)>h_C(B)$ for $B$ large enough. This means $B/h^*(B)<B/h_C(B)=\log(B)/C$ for $B$ large enough and because $x \mapsto x\exp(-x)$ is strictly decreasing with $x$ for large $x$,
\begin{align}
\lambda_1(C^*(B), B)  &= B + \frac{1}{\pi}  \frac{B^2}{h^*(B)} \exp\Big(-\frac{B}{2\pi h^*(B)}\Big)   (1+O((B/h^*(B))^{-1}))   + \frac{\pi^2}{h^*(B)^2} \\
 &\geq B + \dfrac{1}{\pi C}   \log(B) B^{1-\frac{1}{2\pi C}}   (1+o(1)) \label{eq:lambda1_hstar_lower_bound_C1_6pi} 
\end{align}
as $B\to \infty$. But for $C> 1/(6\pi)$ this lower bound is asymptotically larger than the expression given by \eqref{eq:lambda1_Cequal1_6pi} when considering a cylinder with height $h=h_{1/(6\pi)}(B)$, contradicting the optimality of $h^*(B)$. Thus, 
\begin{align}
h^*(B) \leq \frac{1}{6\pi} \frac{B}{\log(B)} (1+o(1)) \label{eq:h_star_lowerbnd}
\end{align}
as $B\to \infty$. 

Next, assume there exists a constant $0<C< 1/(6\pi)$ such that $h^*(B)<h_C(B)$ for $B$ large enough. This means $B/h^*(B)>B/h_C(B)=\log(B)/C$ for $B$ large enough and hence
\begin{align}
\lambda_1(C^*(B), B)  &= B + \frac{1}{\pi}  \frac{B^2}{h^*(B)} \exp\Big(-\frac{B}{2\pi h^*(B)}\Big)   (1+O((B/h^*(B))^{-1}))   + \frac{\pi^2}{h^*(B)^2} \\
 &\geq B + \dfrac{\pi^2}{C^2} \dfrac{\log(B)^2}{B^2} 
\end{align}
But again, for $C< 1/(6\pi)$ this lower bound is asymptotically larger than the expression given by \eqref{eq:lambda1_Cequal1_6pi}, contradicting the optimality of $h^*(B)$. Hence, 
\begin{align}
h^*(B) \geq \frac{1}{6\pi} \frac{B}{\log(B)} (1+o(1)) \label{eq:h_star_upperbnd}
\end{align}
as $B\to \infty$. In view of \eqref{eq:h_star_lowerbnd} and \eqref{eq:h_star_upperbnd}, we conclude \eqref{eq:thm_h_star_6pi}. 

Furthermore, it is clear from \eqref{eq:lambda1_Cequal1_6pi} that
\begin{align}
\lambda_{1,cyl}^*(B) &\leq \lambda_1(C_{R,h_{1/(6\pi)}(B)}, B) = B + 36\pi^4 \dfrac{\log(B)^2}{B^2} (1+o(1))  \label{eq:lambda1_star_upperbnd}
\end{align}
as $B\to \infty$. On the other hand,
\begin{align}
\lambda_{1,cyl}^*(B) &= \lambda_1(C^*(B), B) \geq B   + \frac{\pi^2}{h^*(B)^2} = B + 36\pi^4 \dfrac{\log(B)^2}{B^2} (1+o(1)) \label{eq:lambda1_star_lowerbnd}
\end{align}
as $B\to \infty$. From \eqref{eq:lambda1_star_upperbnd} and \eqref{eq:lambda1_star_lowerbnd} follows \eqref{eq:thm_h_star_6pi_2}.
\end{proof}

The above analysis for cylinders has a number of consequences on possible minimizers $\Omega^*(B)$ of the full shape optimization problem. To formulate the next result, let us define for any bounded, open domain $\Omega \subset \mathbb{R}^3$ the projections 
\begin{align}
\Omega_z := \{ z \, : \, (x,y,z) \in \Omega\}, \qquad  \Omega_{xy}:= \{(x,y) \, : \, (x,y,z) \in \Omega \}
\end{align}
and the geometric quantities
\begin{align}
h(\Omega):=\operatorname{diam}(\Omega_z) \, , \qquad R(\Omega):=\inf_{c \in \mathbb{R}^2} \sup_{x \in \Omega_{xy}} \big( |x_1-c_1|^2+|x_2-c_2|^2 \big)^{1/2} \, .
\end{align}
By definition, $h(\Omega)$ is the length of the projection of $\Omega$ onto the $z$-axis and $R(\Omega)$ is the outer radius of the projection of $\Omega$ to the $xy$-plane. For the next result, we note that the basic estimate $\lambda_1(\Omega,B) \geq B$ from Proposition \ref{prop:simple_est_domain_mono} (2) can be improved by an additional constant that involves $h(\Omega)$.

\begin{lemma} \label{lem:basic_est_pluspih2}
Let $\Omega \subset \mathbb{R}^3$ be a bounded, open domain. Then 
\begin{align}
\lambda_1(\Omega, B) \geq B + \frac{\pi^2}{h(\Omega)^2} , \qquad B\geq 0.
\end{align}
\end{lemma}

\begin{proof}
Since $\Omega$ is bounded, one can choose a large $R>0$ such that $\Omega$ is contained in the cylinder $C_{R,h(\Omega)}$ (after some translation). Then, by Proposition \ref{prop:simple_est_domain_mono} (1), (2) and (4),
\begin{align}
\lambda_1(\Omega, B) \geq \lambda_1(C_{R,h(\Omega)}, B) = \lambda_1(D_R,B)  + \frac{\pi^2}{h(\Omega)^2}  \geq B + \frac{\pi^2}{h(\Omega)^2}, \qquad B\geq 0.
\end{align}
\end{proof}

If we compare this lower bound with the asymptotic expansion of $\lambda_{1,cyl}^*(B)$ from Theorem \ref{thm:lambda1cyl_asympt}, one sees that the first eigenvalue of optimal cylinders is eventually smaller than that of any fixed bounded domain. We conclude

\begin{corollary}
Let $\Omega \subset \mathbb{R}^3$ be a bounded, open domain. Then there exists $B_0(\Omega)>0$ such that $\Omega^*(B)$ cannot be $\Omega$ for any $B \geq B_0(\Omega)$.
\end{corollary} 

and in particular that $\Omega^*(B)$ cannot be a ball for $B$ large enough.

Comparing optimal cylinders with general minimizers $\Omega^*(B)$ by employing the monotonicity principle leads to more precise statements on the behavior of the geometry of $\Omega^*(B)$ as $B$ grows. The following result shows that minimizers of \eqref{eq:min_problem_c} must elongate (or ``spaghettify'') along the direction of the magnetic field.

\begin{corollary} \label{cor:h_Romegastar}
\begin{align}
h(\Omega^*(B)) &\geq \frac{1}{6\pi} \frac{B}{\log(B)} (1+o(1)) \qquad \text{as } B\to \infty .
\end{align}
and
\begin{align}
R(\Omega^*(B))^2 & \geq 6 \frac{\log(B)}{B} (1+o(1)) \qquad \text{as } B\to \infty,
\end{align}

\end{corollary}

\begin{proof}
Suppose there exists a $C<6$ such that $R(\Omega^*(B))^2 < C \log(B) /B$ for all $B$ large enough. Let $C_{R,h}$ be a circular cylinder of radius $R=R(\Omega^*(B))$ and arbitrary height $h$ that contains $\Omega^*(B)$ up to some translation. By Proposition \ref{prop:simple_est_domain_mono} (4), we obtain
\begin{align}
\lambda_1(\Omega^*(B),B) \geq \lambda_1(C_{R,h},B) = \lambda_1(D_R, B)  + \frac{\pi^2}{h^2}  > \lambda_1(D_R, B) 
\end{align}
where
\begin{align}
\lambda_1(D_R, B)&= B + B^2  R(\Omega^*(B))^2 \exp\left( -\frac{B R(\Omega^*(B))^2}{2}\right) (1+O((B R(\Omega^*(B))^2)^{-1})) \\
&\geq B + C \log(B) B^{1- \frac{C}{2}} (1+O((\log(B))^{-1})) 
\end{align}
as $B\to \infty$. But for $C<6$, we then conclude $\lambda_1(\Omega^*(B),B)>\lambda_{1,cyl}^*(B)$ for large enough $B$, contradicting the optimality of $\Omega^*(B)$. Therefore, $R(\Omega^*(B))^2 \geq C \log(B) /B$ for large enough $B$.

Now, suppose there exists $C<1/(6\pi)$ such that $h(\Omega^*(B))< C B/\log(B)$ for all $B$ large enough. Let $C_{R,h}$ be a circular cylinder of height $h=h(\Omega^*(B))$ and arbitrary radius $R$ that contains $\Omega^*(B)$ up to some translation. By Proposition \ref{prop:simple_est_domain_mono} (4) and \eqref{eq:easy_cylinder_estimate}, we obtain
\begin{align}
\lambda_1(\Omega^*(B),B) \geq \lambda_1(C_{R,h},B) \geq B + \frac{\pi^2}{h(\Omega^*(B))^2}> B + \frac{\pi^2}{C^2} \frac{\log(B)^2}{B^2}.
\end{align}
But then, we find $\lambda_1(\Omega^*(B),B) > \lambda_{1,cyl}^*(B)$ for large $B$ which is again contradicting the optimality of $\Omega^*(B)$. Hence, $h(\Omega^*(B))\geq C B/\log(B)$ for large enough $B$.
\end{proof}

Note that since $\operatorname{diam}(\Omega) \geq h(\Omega)$ for any domain $\Omega$, the above corollary also shows 
\begin{align} \label{eq:diam_asympt_lower_bnd}
\operatorname{diam}(\Omega^*(B)) \geq \frac{1}{6\pi} \frac{B}{\log(B)}(1+o(1))
\end{align}
as $B\to \infty$. However, the above corollary is more specific since it states that the minimizer must be elongated along the axis of the magnetic field if the field strength is large.

Under the assumption that $\Omega^*(B)$ is convex, we can also provide a rough upper bound to the growth rate of $\operatorname{diam}(\Omega^*(B))$. For this, we apply the following theorem due to Hersch in two dimensions and extended to dimensions $d\geq 3$ by Payne and Stakgold.

\begin{theorem}[Hersch-Payne-Stakgold \cite{Hersch1960, Payne1973, Avkhadiev2007}]
Let $\Omega\subset \mathbb{R}^d$ be an open, convex set. Then 
\begin{align}
\lambda_1(\Omega,0) \geq  \frac{\pi^2}{4r_{in}(\Omega)^2} .
\end{align}
Here, $r_{in}(\Omega)$ denotes the inradius of $\Omega$. 
\end{theorem}

First and foremost, this result can be used to bound the inradius of $\Omega^*(B)$ from below. Specifically, if $\Omega^*(B)$ is convex, we have
\begin{align}
\lambda_1^*(B) = \lambda_1(\Omega^*(B),B) \geq \lambda_1(\Omega^*(B),0) \geq  \frac{\pi^2}{4r_{in}(\Omega^*(B))^2} 
\end{align}
and therefore
\begin{align} \label{eq:inr_lower_bnd_lambda1star}
r_{in}(\Omega^*(B))^2 \geq \frac{\pi^2}{4\lambda_1^*(B)}. 
\end{align}

Next, we relate $r_{in}(\Omega^*(B))$ and $\operatorname{diam}(\Omega^*(B))$ using the following lemma.

\begin{lemma} \label{lem:convex_inr_diam_bnd}
Let $\Omega\subset \mathbb{R}^3$ be an open, convex set. Then 
\begin{align}
|\Omega| \geq  \frac{\pi}{3} r_{in}(\Omega)^2 \cdot \operatorname{diam}(\Omega). \label{convex_inr_diam_bnd}
\end{align}
\end{lemma}

The statement of the lemma follows for example from a result of Delyon, Henrot and Privat \cite{Delyon2022} which characterizes convex shapes that minimize the $d$-dimensional volume under given inradius and diameter. However, Lemma \ref{lem:convex_inr_diam_bnd} can be proven by elementary means, so a proof shall be given here for the convenience of the reader.

\begin{proof}
Let $\varepsilon>0$ be sufficiently small. By definition of the diameter, there exist two points $x,y\in\Omega$ with $|x-y|>\operatorname{diam}(\Omega) - \varepsilon$ and by definition of the inradius, there is a $z\in \Omega$ such that $\mathbb{B}_{r_{in}(\Omega)- \varepsilon}(z)= \{x \in \mathbb{R}^3 \, : \, |x-z|< r_{in}(\Omega)- \varepsilon \} \subset \Omega$.

Consider the plane $P$ that contains $z$ and is orthogonal to the line that goes through $x$ and $y$. The convex hull of the flat disk $D_\varepsilon = P \cap \mathbb{B}_{r_{in}(\Omega)- \varepsilon}(z)$ together with the point $x$ forms an oblique cone $C_1$ with base $D_\varepsilon$ and apex $x$. In the same way, another oblique cone $C_2$ is yield with base $D_\varepsilon$ and apex $y$. Since $\Omega$ is convex, one has $C_1 \cup C_2 \subset \Omega$. 

If $P$ separates $x$ and $y$, then $|x-y|=\operatorname{dist}(x,P)+\operatorname{dist}(y,P)$ and thus 
\begin{align}
|\Omega| \geq |C_1 \cup C_2| = |C_1| + |C_2| &= \frac{\pi}{3} ( r_{in}(\Omega)- \varepsilon )^2 \cdot \operatorname{dist}(x,P) + \frac{\pi}{3} ( r_{in}(\Omega)- \varepsilon )^2 \cdot \operatorname{dist}(y,P) \\
& = \frac{\pi}{3} ( r_{in}(\Omega)- \varepsilon )^2 \cdot |x-y|\\
& > \frac{\pi}{3} ( r_{in}(\Omega)- \varepsilon )^2 \cdot (\operatorname{diam}(\Omega) - \varepsilon).
\end{align}
Otherwise $\operatorname{dist}(x,P) \geq |x-y|$ or $\operatorname{dist}(y,P) \geq |x-y|$ and 
\begin{align}
|\Omega| \geq \max\{ |C_1| , |C_2| \} & \geq \frac{\pi}{3} ( r_{in}(\Omega)- \varepsilon )^2 \cdot |x-y|\\
& > \frac{\pi}{3} ( r_{in}(\Omega)- \varepsilon )^2 \cdot (\operatorname{diam}(\Omega) - \varepsilon).
\end{align}

Since $\varepsilon$ was arbitrary, the desired inequality follows.
\end{proof}

We remark that the constant $\pi/3$ is sharp and cannot be improved. This can be seen by taking $\Omega$ as the interior of the convex hull of two very distant points $x$ and $y$ and a small ball centered at the midpoint between the two points $x$ and $y$. The bound \eqref{convex_inr_diam_bnd} is saturated when the radius of the ball is kept fixed and the distance of $x$ and $y$ sent to infinity.

As $\Omega^*(B)$ is subject to the volume constraint $ |\Omega^*(B)|=1$, Lemma \ref{lem:convex_inr_diam_bnd} together with \eqref{eq:inr_lower_bnd_lambda1star} implies that
\begin{align}
\operatorname{diam}(\Omega^*(B)) \leq \frac{12}{\pi^3} \lambda_1^*(B).
\end{align}
Recall that $B \leq \lambda_1^*(B)\leq \lambda_{1,cyl}^*(B)$, thus $\lambda_1^*(B) = B(1+o(1))$ as $B\to\infty$. Hence, we have shown that the diameter of the minimizers can grow at most linearly.

Similarly, Lemma \ref{lem:convex_inr_diam_bnd} implies together with \eqref{eq:diam_asympt_lower_bnd} that
\begin{align}
r_{in}(\Omega^*(B))^2 \leq 18 \frac{\log(B)}{B} (1+o(1))
\end{align} 
as $B\to\infty$. This means that $r_{in}(\Omega^*(B))^2$ must vanish at least with rate $\log(B)/B$. Meanwhile, inequality \eqref{eq:inr_lower_bnd_lambda1star} shows that $r_{in}(\Omega^*(B))^2$ cannot vanish with a rate faster than $1/B$.

Thus, we have proven under convexity assumption on the minimizers $\Omega^*(B)$ the following two-sided bounds.

\begin{corollary} \label{cor:convex_degen}
Suppose there exists $B_0>1$ such that $\Omega^*(B)$ is convex for all $B\geq B_0$, then 
\begin{align}
\frac{B}{\log(B)} \, &\lesssim \, \operatorname{diam}(\Omega^*(B)) \,\lesssim \quad B, \\
\frac{1}{B} \quad &\lesssim \; r_{in}(\Omega^*(B))^2 \;\, \lesssim \frac{\log(B)}{B}
\end{align}
for all $B\geq B_0$. Here, we write $f(x) \lesssim g(x)$, if there exists a constant $c>0$ such that $f(x) \leq c \cdot g(x)$ for all $x$.
\end{corollary}

To conclude this section, we want to comment briefly on the Neumann case. While for Dirichlet boundary conditions one obtains ``spaghettification'' of optimizers, one observes the opposite phenomenon in the case of the reversed shape optimization problem for Neumann boundary conditions. 

Let $\mu_n(\Omega, B)$ denote the eigenvalues of the Neumann realization of the magnetic Laplacian with constant magnetic field on $\Omega \subset \mathbb{R}^d$, $d=2$ or $3$, and consider the corresponding shape optimization problem
\begin{align}
\max_{\substack{\Omega\subset \mathbb{R}^3 \text{ open}, \\ |\Omega|=c}} \mu_1(\Omega, B). \label{eq:min_problem_neumann}
\end{align}
For $B\to 0$, any possible minimizer of \eqref{eq:min_problem_neumann} must widen in the plane orthogonal to the direction of the magnetic field (we are tempted to call this behavior ``pancaking''). 

The following argument we provide is based on the observation by Fournais and Helffer \cite[Section 5.3]{Fournais2019} that the ball cannot be the minimizer of \eqref{eq:min_problem_neumann} if $B$ is fixed and $c$ is small enough, or equivalently, if $c$ is fixed and the field strength $B$ is small enough. The argument is a simple comparison with very flat cylinders.

By taking the constant trial function $u(x)=|\Omega|^{-1/2}$, one has for the first Neumann eigenvalue of any domain $\Omega \subset \mathbb{R}^3$ the elementary upper bound
\begin{align}
\mu_1(\Omega,B) \leq \int_\Omega |A|^2 |\Omega|^{-1} \dd x = \frac{B^2}{4|\Omega|} \int_\Omega x_1^2+x_2^2 \dd x.
\end{align}
Since $\mu_1(\Omega,B)$ is translation invariant, we can replace $\Omega$ by any translation of it, giving the improved bound
\begin{align}
\mu_1(\Omega,B) \leq \frac{B^2}{4|\Omega|} \inf_{a\in\mathbb{R}^3} \int_{\Omega+a} x_1^2+x_2^2 \, \dd x.
\end{align}
A simple estimate for the infimum is 
\begin{align}
\inf_{a\in\mathbb{R}^3} \int_{\Omega+a} x_1^2+x_2^2 \, \dd x \leq  |\Omega| R(\Omega)^2,
\end{align}
hence giving
\begin{align}
\mu_1(\Omega,B) \leq \frac{B^2}{4} R(\Omega)^2.
\end{align}
On the other hand, consider a cylinder $C_{R,h}$. By separation of variables,
\begin{align}
\mu_1(C_{R,h},B) = \mu_1(D_R,B) = h \mu_1 \bigg(D, \frac{B}{h} \bigg) = B \cdot \frac{\mu_1(D, Bh^{-1})}{Bh^{-1}}
\end{align}
and it is known that
\begin{align}
\frac{\mu_1(D, Bh^{-1})}{Bh^{-1}} \to \Theta_0 \approx 0.59
\end{align}
as $h\to 0$, where $\Theta_0$ is the so-called \textit{De Gennes constant}, see \cite[Theorem 2.5]{Fournais2006a}. One concludes that if 
\begin{align}
\frac{B^2}{4} R(\Omega)^2 <  \Theta_0 B,
\end{align}
then $\Omega$ cannot be a maximizer of \eqref{eq:min_problem_neumann}. Thus, any maximizer $\Omega$ for \eqref{eq:min_problem_neumann} must satisfy
\begin{align}
R(\Omega)^2 \geq \frac{ 4 \Theta_0 }{B}.
\end{align}
This shows that as $B\to 0$, maximizers for \eqref{eq:min_problem_neumann} must widen in the plane orthogonal to the direction of the magnetic field.

\section{Numerical Methods} \label{sec:nummeth}

In this section, we discuss the methods we applied to numerically solve the shape optimization problem \eqref{eq:min_problem_c}. We begin by describing our choice of domain parametrization. We then give a brief overview of the Method of Particular Solutions which is the method we chose to compute eigenvalues. Then, we explain modifications for axisymmetric domains that improve efficiency of the eigenvalue solver. Finally, we summarize the gradient descent scheme used for shape optimization and the derivative formulas that are needed as part of the gradient computations. 

\subsection{Domain parametrization} \label{subsec:domainparam}

As proposed in \cite{Antunes2017}, we focus for the numerical optimization routine on star-shaped domains $\Omega \subset \mathbb{R}^3$ that can be described by a variable radius function $r(\theta,\varphi)$ which depends on two angles $\theta$ and $\varphi$ in spherical coordinates. More precisely, we assume domains $\Omega$ are of the form
\begin{align}
\Omega = \{ (r \sin(\theta) \cos(\varphi),r \sin(\theta) \sin(\varphi),r \cos(\theta))^T \in \mathbb{R}^3 \, : \, 0\leq r < r(\theta,\varphi), \theta \in [0,\pi], \varphi \in [0,2\pi)  \}.
\end{align}
where $r(\theta,\varphi)$ is a positive and smooth function of the angles. We define
\begin{align} \label{eq:star_shaped_set3d}
\gamma(\theta,\varphi):= r(\theta,\varphi) \begin{pmatrix} \sin \theta \cos \varphi \\  \sin \theta \sin \varphi \\  \cos \theta \end{pmatrix} , \qquad \theta \in [0,\pi], \, \varphi \in [0,2\pi).
\end{align}
The image of $\gamma$ is the boundary surface of $\Omega$, i.e.
$$\partial \Omega = \{\gamma(\theta,\varphi) \, : \,  \theta \in [0,\pi], \varphi \in [0,2\pi)  \}.$$
We further assume that the radius function $r(\theta,\varphi)$ is given by a truncated series of real spherical harmonics
\begin{align}
r(\theta,\varphi) = \sum\limits_{l=0}^{l_{max}} \sum\limits_{m=-l}^{l} c_{l,m} \tilde{Y}_{l}^{m}(\theta, \varphi)
\end{align}
where
\begin{align}
\tilde{Y}_{l}^{m}(\theta, \phi) = \begin{cases}
\sqrt{2} \, \mathrm{Re} \,Y_{l}^{m}(\theta, \varphi) & \text{ if } m>0\\
\quad \;\; \mathrm{Re} \,Y_{l}^{m}(\theta, \varphi) & \text{ if } m=0\\
\sqrt{2}  \, \mathrm{Im}\, Y_{l}^{|m|}(\theta, \varphi) & \text{ if } m<0\\
\end{cases}
\end{align}
and
\begin{align}
Y_{l}^{m}(\theta, \varphi) = \sqrt{\frac{2l+1}{4\pi} \frac{(l-m)!}{(l+m)!}}e^{im \varphi} P_{l}^{m}(\cos \theta), \qquad \qquad l\geq m\geq 0,
\end{align}
are the usual spherical harmonics, with $P_{l}^{m}$ denoting associated Legendre functions. The total number of coefficients $c_{l,m}$ that determine the boundary and thus the shape of the domain is
\begin{align}
n_c = \sum\limits_{l=0}^{l_{max}} (2l+1) = (l_{max} + 1)^2.
\end{align}
Collecting the coefficients $c_{l,m}$ in a vector
\begin{align}
c = (c_{0,0}, c_{1,-1}, c_{1,0}, c_{1,1}, c_{2,-2},..., c_{l_{max},l_{max}})^T \in \mathbb{R}^{n_c}
\end{align}
and likewise collecting the real spherical harmonics in a vector 
\begin{align}
\tilde{Y}(\theta, \varphi) = (\tilde{Y}_{0}^{0}(\theta, \varphi), \tilde{Y}_{1}^{-1}(\theta, \varphi), \tilde{Y}_{1}^{0}(\theta, \varphi), \tilde{Y}_{1}^{1}(\theta, \varphi), \tilde{Y}_{2}^{-2}(\theta, \varphi), ..., \tilde{Y}_{l_{max}}^{l_{max}}(\theta, \varphi))^T,
\end{align}
we can write
\begin{align}
r(\theta,\varphi) = c^T \cdot \tilde{Y}(\theta, \varphi)
\end{align}
where $\cdot$ is the standard Euclidean dot product. For numerical computations, we typically use $l_{max}=10$ which corresponds to $n_c = 121$ coefficients.

For domains $\Omega \subset \mathbb{R}^3$ that are axisymmetric around the $z$-axis, we use a simplified representation for the radius function $r(\theta, \varphi)$, since all coefficients $c_{l,m}$ with $m\neq 0$ must be zero. In this case, we write 
\begin{align}
r(\theta,\varphi) = \sum\limits_{l=0}^{l_{max}^{rot}}  c_{l,0} \tilde{Y}_{l}^{0}(\theta, \varphi) = (c^{rot})^T \cdot \tilde{Y}^{rot}(\theta, \varphi)
\end{align}
where 
\begin{align}
c^{rot} &= (c_{0,0}, c_{1,0}, c_{2,0},..., c_{l_{max}^{rot},0})^T \in \mathbb{R}^{n_c^{rot}}, \qquad n_c^{rot}= l_{max}^{rot}+1, \\
\tilde{Y}^{rot}(\theta, \varphi) &= (\tilde{Y}_{0}^{0}(\theta, \varphi),  \tilde{Y}_{1}^{0}(\theta, \varphi),  \tilde{Y}_{2}^{0}(\theta, \varphi), ..., \tilde{Y}_{l_{max}^{rot}}^{0}(\theta, \varphi))^T.
\end{align}
Typically, we use $l_{max}^{rot} = 60$ for parametrization of axisymmetric domains.

\subsection{Method of Particular Solutions} \label{subsec:mps}

The Method of Particular Solutions \cite{Fox1967,Betcke2005} is a meshless method for numerical solution of partial differential equations that has been getting renewed attention in recent years. We briefly give the main idea of the method in the following.

We want to solve an eigenvalue problem of the form
\begin{align}
H u &= \lambda u \qquad \text{in } \Omega, \label{eq:mps_eigequ} \\ 
u &= 0 \qquad \text{on } \partial \Omega \label{eq:mps_bndcond}
\end{align}
where $H$ is some partial differential operator. Suppose we have a set of functions $\{\psi_j(\, \cdot \, ;\lambda)\}_{j=1}^{m}$ which satisfy the eigenvalue equation \eqref{eq:mps_eigequ}, i.e.\ we assume that 
\begin{align}
H \psi_j(\, \cdot \, ;\lambda) = \lambda \psi_j(\, \cdot \, ;\lambda) \qquad \text{in } \Omega
\end{align}
for any $j$. Any function spanned by these particular solutions, i.e.\ any function $u$ given by
\begin{align}
u_\alpha(x) = \sum_{j=1}^m \alpha_j \psi_j(x ;\lambda) \label{eq:mps_u_from_span}
\end{align}
where $\alpha=(\alpha_1, ..., \alpha_m)$ is an arbitrary vector of complex coefficients, also satisfies $Hu_\alpha = \lambda u_\alpha$ in $\Omega$. However, in general, such a function will not satisfy the boundary condition \eqref{eq:mps_bndcond}. We now try to solve \eqref{eq:mps_bndcond} within the span of our basis functions by imposing the boundary condition \eqref{eq:mps_bndcond} in a finite set of collocation points $\{ x_i \} \subset \partial \Omega$. This leads to the linear system of equations
\begin{align}
A(\lambda) \cdot \alpha = 0 \label{eq:A_lambda}
\end{align}
where $A(\lambda)$ is the matrix given by $(A(\lambda))_{ij} = \psi_j(x_i ;\lambda)$ and $\alpha = (\alpha_1, ..., \alpha_m)^T$. 

We then employ a method called Subspace Angle Technique to search for non-trivial pairs $(\hat{\lambda}, \hat{\alpha})$ where \eqref{eq:A_lambda} is (approximately) fulfilled. The function $u_{\hat{\alpha}}$ is then an approximate eigenfunction to the approximate eigenvalue $\hat{\lambda}$. We refer to Betcke and Trefethen \cite{Betcke2005} for details on the Subspace Angle Technique. \\

\textbf{Particular solutions of the magnetic Laplacian in $\mathbb{R}^3$} \\

In our specific problem of the magnetic Laplacian, we have $H=(-i\nabla+ A)^2$ with $A(x_1,x_2,x_3)=B/2(-x_2,x_1,0)^T$. We assume in the following that $B>0$. We choose particular solutions in cylindrical coordinates as
\begin{align}
\psi_{l,p}(r,\varphi,z; \lambda) &= e^{-\frac{\phi	}{2}} \phi^{\frac{|l|}{2}}  M\left(a,b, \phi \right) e^{il \varphi} e^{ipz}
\end{align}
where
\begin{align}
\phi = \frac{Br^2}{2}, \qquad a = \frac{1}{2}\Big(l+|l|+1 - \frac{\lambda-p^2}{B}\Big), \qquad b=|l|+1.
\end{align}
and $M(a,b,z)$ denotes Kummer's confluent hypergeometric function. It is somewhat tedious but straightforward to check that indeed $H\psi_{l,p} = \lambda \psi_{l,p}$ on $\mathbb{R}^3$. Notice that the basis functions are indexed by two quantum numbers, $l$ and $p$, which can seen as the angular momentum around the $z$-axis and the momentum in $z$-direction. Equation \eqref{eq:mps_u_from_span} now becomes
\begin{align}
u_\alpha(r,\varphi,z) = \sum_{l,p\in L,P } \alpha_{l,p}\psi_{l,p}(r,\varphi,z; \lambda) \label{eq:basis_funct_magn}
\end{align}
where $L = \{-n_l,  ..., n_l\}$, $P = \delta p \cdot \{-n_p,  ..., n_p\}$. For our numerical computations, we chose 
\begin{align} \label{eq:MPS_parameters}
n_l = 10, \quad n_p = 8, \quad  \delta p = \frac{\sqrt{\lambda}}{n_p+1}.
\end{align}

\begin{remark}
If $H=-\Delta$ is the standard (non-magnetic) Laplacian, then one can choose the particular solutions
\begin{align*}
\psi_{l,p}(r,\varphi,z; \lambda) = J_{|l|}\big(\sqrt{\lambda-p^2} \,r \big) e^{il\varphi} e^{ipz}
\end{align*}
in cylindrical coordinates. Here, $J_{\nu}$ denotes the Bessel function of order $\nu$. \\
\end{remark}

\textbf{Collocation points}\\

We set collocation points according to $x_i = \gamma(\theta_i, \varphi_i) $ where $\gamma$ is the surface map from \eqref{eq:star_shaped_set3d} and the angles $(\theta_i, \varphi_i)$ are chosen from a finite set of angles $M_{\theta,\varphi}$. We set up $M_{\theta,\varphi}$ by the following procedure. 

Suppose we have a target number of collocation points $N_{target}$. In our computations, we typically set  $N_{target}=1000$. First, set
\begin{align}
n_\theta = \bigg\lceil \frac{1}{2}\sqrt{\pi N_{target}} \bigg\rceil. \label{eq:n_theta}
\end{align}
We then choose $n_\theta+1$ equidistant angles $\theta_j \in [0,\pi]$ according to
\begin{align}
\theta_j &= \frac{\pi j}{n_\theta} , \qquad j=0,..., n_\theta.
\end{align}
The angles $\theta_0=0$ and $\theta_{n_\theta}=\pi$ correspond to points on the surface of the domain $\Omega$ that intersect with the $z$-axis. For a sphere, these points would be its poles. The other angles $\theta_j$ would correspond to $n_\theta-1$ evenly spaced points on circles of latitude. In the next step, we set for every latitude $\theta_j$ 
\begin{align}
n_{\varphi,j} &= \begin{cases} 
1, & j=0 \text{ or } n_\theta, \\[1ex] 
\bigg\lfloor 2 n_\theta \sin(\theta_j) \bigg\rfloor, &  \text{else}.\\
\end{cases} 
\end{align}
We then divide each circle of latitude evenly by setting 
\begin{align}
\varphi_{j,k} &=\frac{2\pi k}{n_{\varphi,j}}  , \qquad j=0,..., n_\theta, \qquad  k=0,..., n_{\varphi,j}-1,
\end{align}
Finally, the set of latitudes and longitudes for the collocation points is 
\begin{align}
M_{\theta,\varphi}= \{ (\theta_j, \varphi_{j,k}) \, : \, j=0,..., n_\theta, \, k=0,..., n_{\varphi,j}-1 \}.
\end{align}
Notice that the method we described produces 
\begin{align}
\# M_{\theta,\varphi} = \sum_{j=0}^{n_\theta} n_{\varphi,j} \approx  \sum_{j=0}^{n_\theta} 2n_{\theta} \sin(\theta_j) \approx N_{target}
\end{align}
collocation points.

%
%
%

\subsection{Modifications for axisymmetric domains} \label{subsec:modaxis}

Numerical computations for certain domains $\Omega$ that are axisymmetric with respect to the $z$-axis (for example balls, ellipsoids and cylinders aligned with the $z$-axis) suggest that the first eigenvalue of the magnetic Dirichlet Laplacian on such domains is simple for any $B\geq 0$. Furthermore, it appears that the associated ground state eigenfunction is also axisymmetric with respect to the $z$-axis, i.e.\ independent of the angular coordinate $\varphi$.

Assuming that the ground state eigenfunction is always axisymmetric, we can incorporate this symmetry in our ansatz for the MPS eigenfunction. We replace the \eqref{eq:basis_funct_magn} by the simplified ansatz
\begin{align}
u_\alpha(r,\varphi,z) = \sum_{p\in P } \alpha_{p}\psi_{0,p}(r,\varphi,z; \lambda),
\end{align}
where we only use particular solutions with $l=0$, i.e.\ solutions that are axisymmetric with respect to the $z$-axis. This way, the MPS ansatz $u_\alpha(r,\varphi,z)$ becomes axisymmetric with respect to the $z$-axis for any choice of coefficients $\alpha$. 

As $u_\alpha(r,\varphi,z)$ is now independent of $\varphi$, we change our choice of collocation points accordingly. It is not necessary to choose spherical angles $(\theta_j, \varphi_{j,k})$ from the set $M_{\theta,\varphi}$ with different longitude $\varphi_{j,k}$. It rather suffices to choose collocation points from the line $\gamma(\theta, 0)$ on the surface of $\Omega$ to enforce the collocation condition on whole circles of latitude. We thus replace $M_{\theta,\varphi}$ from the previous section by 
\begin{align}
M_{\theta,\varphi}= \{ (\theta_j, 0) \, : \, j=0,..., n_\theta \}.
\end{align}

We point out that under the axisymmetry assumption on the ground state, the changes laid out above lead to a great reduction of the size of the MPS matrix $A(\lambda)$ compared to the general MPS ansatz, while preserving all collocation conditions (indeed, they are even extended to whole circles of latitude). Instead of approximately $N_{target}$ collocation points, we only have about $\sqrt{N_{target}}$ collocation points and instead of $(2n_l+1)\cdot (2n_p+1)$ particular solutions, we only employ $(2n_p+1)$ particular solutions. With the same numbers $n_\theta$, $n_p$ as in the general ansatz, we have therefore a much more efficient eigenvalue solver. At the same time, we can expect the ground state eigenpair approximations to be of similar accuracy.

The greater efficiency of above described eigenvalue solver for axisymmetric domains can be traded for greater accuracy. Thus we increase the number of particular solutions and collocation points. This increases the matrix size again. So for the axisymmetric case, instead of \eqref{eq:MPS_parameters} and \eqref{eq:n_theta}, we use the parameters
\begin{align}
n_\theta = 1000, \qquad n_l=10, \qquad n_p = 60, \qquad \delta p =  \frac{10\sqrt{\lambda}}{n_p+1},
\end{align}
which leads to the matrix $A(\lambda)$ having similar size as in the general non-symmetric MPS approach.

Finally, we want to spend a few words on the assumption of axisymmetric ground states of the magnetic Dirichlet Laplacian on axisymmetric domains.

Since the ground state of the non-magnetic Dirichlet Laplacian is simple and admits a positive, axisymmetric eigenfunction, the same follows for small $B$ by a perturbation argument. For cylinders aligned with the $z$-axis, it is a fact that the ground state is axisymmetric for any $B\geq 0$. This is an immediate consequence of the fact that variables separate in cylindrical coordinates and that the ground state on disks is always simple and rotationally symmetric (see \cite{Son2014} for a proof). For balls, the situation is much less clear, since variables do not fully separate in cylindrical coordinates. To our knowledge it is an open problem whether the ground state on a ball is indeed simple and axisymmetric for any $B\geq 0$. We note that simplicity implies axisymmetry on the ball.

We remark that simplicity as well as axisymmetry of the ground state can fail for large $B$ if there are no further restrictions on the geometry of $\Omega$ than axisymmetry. For example, take $\Omega$ as a cylindrical shell aligned with the $z$-axis, so a cartesian product of an annulus and an interval. By separation of variables, the ground state on $\Omega$ is given by the product of the ground state on the annulus and the ground state of the interval. Solving the eigenvalue problem of the annulus in polar coordinates, one quickly realizes that one can tune the dimensions of the annulus and the field strength $B>0$, so that the ground state exhibits a non-zero angular momentum. More precisely, the eigenfunction on the annulus has a factor $e^{il\varphi}$ with $l<0$. The ground state on $\Omega$ is thus not axisymmetric. Moreover, recalling that the ground state on $\Omega$ must be axisymmetric for small $B$, analytic perturbation theory implies that there is an intermediate field strength where the ground state on $\Omega$ is degenerate. In that case, one finds both an axisymmetric ground state and a non-axisymmetric ground state associated with the same ground state energy.

Pinpointing sufficient geometric conditions (e.g.~convexity) under which the ground state of an axisymmetric domain $\Omega$ is simple and axisymmetric is outside the scope of this article. We leave it as an open problem to determine a class of domains (a class that preferably includes cylinders and balls) which have axisymmetric ground states for any field strength.

\subsection{Optimization routine}

With the domain parametrization and the eigenvalue solver at hand, we are now able to describe the optimization routine. We begin by noting that \eqref{eq:min_problem_c} is a constrained optimization problem. However, by scaling, one finds that \eqref{eq:min_problem_c} is equivalent to
\begin{align}
\min_{\Omega \text{ open}} |\Omega|^\frac{2}{3} \lambda_1\left(\Omega, \frac{B}{|\Omega|^\frac{2}{3}} \right) \label{eq:min_problem_B_scale_with_c}
\end{align}
in the sense that any minimizer of \eqref{eq:min_problem_c} is a minimizer of \eqref{eq:min_problem_B_scale_with_c} and any minimizer of \eqref{eq:min_problem_B_scale_with_c} scaled to measure one is a minimizer of \eqref{eq:min_problem_c}. Moreover, the minima of both problems are equal. We call $|\Omega|^\frac{2}{3} \lambda_1 (\Omega, B/|\Omega|^\frac{2}{3} )$ the rescaled ground state energy. Problem \eqref{eq:min_problem_B_scale_with_c} has the advantage that it is an unconstrained optimization problem, so we continue with this form of the optimization problem.

The domain parametrization introduced earlier is essentially a map $c \mapsto \Omega_c$ that maps coefficient vectors $c$ to star-shaped domains $\Omega_c$. A coefficient vector $c$ is only admissible if the associated function $r(\theta,\varphi)$ is positive. We now want to optimize the rescaled ground state energy among all domains $\Omega$ that are represented by an admissible coefficient vector $c$. This means that we consider the finite dimensional optimization problem 
\begin{align}
\min_{ \substack{c\in \mathbb{R}^{n_c}, \\ \text{ admissible}} } J(c), \qquad J(c) = |\Omega_c|^\frac{2}{3} \lambda_n\left(\Omega_c, \frac{B}{|\Omega_c|^\frac{2}{3}} \right).
\end{align}
To solve this problem, we apply a modified gradient descent method.\\

\textbf{Modified gradient descent scheme} \\

We adopt the gradient descent scheme from \cite{Antunes2012, Baur2025a}. We start by choosing an initial coefficient vector $c_0$ representing an initial domain $\Omega_{c_0}$. We then compute iteratively the gradient $d=\nabla_c J(c_i)$ for the current domain $\Omega_{c_i}$ and minimize $J$ along the half-line spanned by $-d$. To prevent shrinking or blow-up of the domain associated with $c$, we normalize the coefficient vectors immediately, so that the corresponding domain has unit volume. Thus, instead of searching as usual for the minimum of $J(c_i - \beta d)$, $\beta \geq 0$, we search for the minimum of $J(c(\beta))$ where $c(\beta) = \frac{c_i - \beta d}{|\Omega_{c_i - \beta  d}|^{1/3}}$, $\beta \geq 0$. This one-dimensional optimization is performed with a modified Golden-section search which we motivate and explain briefly in the following.

The objective function $J(c(\beta))$ may have multiple local minima on the half-line $\beta \geq 0$, but since we know that $J(c(\beta))$ decreases near $\beta=0$, it is reasonable to search close to $\beta=0$ for the nearest minimum. Additionally, one does not want to use too large step sizes in the gradient descent scheme as this might produce non-admissible domains with $r(\theta,\varphi)<0$ for some angles $(\theta,\varphi)$. So we set a maximum step-size $\beta_{\mathrm{max}}>0$ (in our implementation we choose $\beta_{\mathrm{max}} = 3 \cdot 10^{-3}$) and an initial guess $\beta_{\mathrm{init}} = 10^{-3} \cdot \| d \|^{-1}$ (in practice, this turns out to be much smaller than $\beta_{\mathrm{max}}$). Our modified Golden-section search starts with a pre-search phase where we select the interval boundaries for the usual Golden-section search. The left interval boundary is chosen to be $\beta_l=0$. To choose the right interval boundary, we start with $\beta_r = \beta_{\mathrm{init}}$ and double the right interval boundary $ \beta_r$ as long as $\beta_r < \beta_{\mathrm{max}} $  and the objective value at  $ \beta_r$ is smaller than the objective value at the mid-point $\beta_m = \beta_r (1-\phi^{-1})$ where $\phi$ denotes the golden ratio.  The stopping criterion with the objective value at the mid-point is such that performing a subsequent Golden-section search on $[\beta_l, \beta_r]$ will start to search in the interior of the interval, away from the right interval boundary $\beta_r$. The choice of $\beta_{\mathrm{init}}$ is deliberately kept very small, so that we are relatively safe to assume that the objective function $J(c(\beta))$ is initially decreasing on $[\beta_l, \beta_r]= [0,\beta_{\mathrm{init}} ]$. In the pre-search phase, the search interval is then enlarged in order to prevent too small gradient descent steps which would mean that the gradient descent scheme converges very slowly. The pre-search phase usually takes only a few doubling iterations to select the search interval $[\beta_l, \beta_r]$. After having selected a search interval, we perform a standard Golden-section search on $[\beta_l, \beta_r]$.

The full shape optimization scheme stops when a certain number of iterations is exceeded or the objective function appears to have converged. The algorithm is summarized below.

{\centering
\begin{minipage}{.8\linewidth}

\begin{algorithm}[H]
\caption{Gradient descent algorithm}
\begin{algorithmic}
\State set $i=0$
\While{$i\leq i_{\mathrm{max}}$ or $|J(c_i)-J(c_{i-1})| \geq \varepsilon$}
	\State compute $J(c_i)$ \hspace*{3.95cm} (Method of Particular Solutions)
	\State compute $d = \nabla_c J(c_i)$
	\State set $c(\beta) = \dfrac{c_i - \beta d}{|\Omega_{c_i - \beta  d}|^{1/3}}$ 
	\State compute $\beta^*= \mathrm{argmin}_{\beta\in [0, \beta_{\mathrm{max}}]} \, J(c(\beta))$ \quad (modified Golden-section search)
	\State set $c_{i+1} = c(\beta^*)$
	\State increase $i \rightarrow i+1$
\EndWhile
\State \Return $J(c_i), c_i$
\end{algorithmic}
\end{algorithm}
\end{minipage}
\par\vspace{0.5cm}
}

It remains to compute the gradient $\nabla_c J(c)$. By chain rule, it is clear that the gradient $\nabla_c J(c)$ is a combination of the three basic derivatives
\begin{align} \label{eq:gradient_chain_rule_comp}
\nabla_c (|\Omega_c|), \qquad \nabla_c \lambda_1\left( \Omega_c, B\right), \qquad \frac{\partial}{\partial B}  \lambda_1\left( \Omega_c, B\right).
\end{align}
All three derivatives can be approximated using boundary/volume integrals and the numerical eigenpair $(\hat{\lambda},u_{\hat{\alpha}})$ from the eigenvalue solver. We use the formulas given in the next section.\\

\textbf{Shape derivatives and Hellmann-Feynman formula} \\

The gradients with respect to domain coefficients can be computed with shape derivative formulas. We recall the basic notion of deformation fields.

Let $\Phi: (-T,T) \to W^{1,\infty}(\mathbb{R}^d; \mathbb{R}^d)$ be a differentiable map such that $\Phi(0)$ = I and $\Phi'(0) = V$ where
$V$ is a vector field. In this context, $V$ is also called a ``deformation field''. For simplicity, we can assume that $\Phi(t) = I + tV$. Applying the map $\Phi$ then generates the family of domains
\begin{align}
\Omega_t = \Phi(t)(\Omega) = \{ \Phi(t)(x) : x \in \Omega \}, \qquad t \in (-T, T).
\end{align}
Note that $\Omega_0 = \Omega$ is the original domain and the parameter $t$ controls the strength of deformation of $\Omega$ by $V$. 

We then have the following shape derivative results.

\begin{theorem}[\cite{Henrot2006}] \label{thm:change_volume}
Let $\Omega \subset \mathbb{R}^d$ be a bounded, open domain. If $\Omega$ is Lipschitz, then $t \mapsto |\Omega_t|$ is differentiable at $t=0$ with
\begin{align}
(|\Omega_t|)'(0) = \int_{\partial \Omega}  V \cdot n \, d\sigma .
\end{align}
\end{theorem}

\begin{theorem}[Hadamard formula for magnetic Dirichlet Laplacian]\label{thm:change_eigval}
Let $\Omega \subset \mathbb{R}^d$, $d=2,3$, be a bounded, open domain with smooth boundary. If $\lambda_n(\Omega,B)$ is a simple eigenvalue with corresponding normalized eigenfunction $u_n \in L^2(\Omega)$, then $\lambda_n(t):=\lambda_n(\Omega_t,B)$ is differentiable at $t=0$ with
$$\lambda_n'(0) = -\int_{\partial \Omega} \left| \frac{\partial u_n}{\partial n}\right|^2 V \cdot n \, d\sigma .$$
Here, $\partial / \partial n$ denotes the normal derivative with respect to the outer normal of the domain. 
\end{theorem}

The proof of the magnetic Hadamard formula is essentially a copy of the proof in the non-magnetic case, see \cite{Henrot2018}. Note that the formula is unchanged compared to the non-magnetic case due to the fact that $\operatorname{div} A =0$ for the standard linear vector potential $A$.

Derivatives of eigenvalues with respect to the field strength $B$ can be computed with a Hellmann-Feynman formula.

\begin{theorem}[Hellmann-Feynman] \label{thm:Hellmann-Feynman}
Let $\Omega \subset \mathbb{R}^d$ be a bounded, open domain with $\partial \Omega$ of class $C^2$ and $\hat{A}: \mathbb{R}^d \rightarrow \mathbb{R}^d$ a sufficiently regular vector field. Denote the eigenvalues of $H(B)=(-i\nabla +B\hat{A})^2$ on $\Omega \subset \mathbb{R}^d$ with Dirichlet boundary condition by $\lambda_n(\Omega,B)$. If $\lambda_n(\Omega,B)$ is a simple eigenvalue with associated $L^2$-normalized eigenfunction $u_n$, then there exists an open neighborhood around $B$ where $\lambda_n(\Omega,B)$ is (infinitely) differentiable with respect to $B$ and where
\begin{align}
\frac{\partial}{\partial B} \lambda_n (\Omega, B) = \left\langle u_n , \frac{\partial H(B)}{\partial B} u_n \right\rangle_{L^2(\Omega)}.
\end{align}

\end{theorem}

If $\hat{A}$ is divergence-free, then the Hellmann-Feynman formula yields
\begin{align} \label{eq:lambda_n_partialB}
\frac{\partial}{\partial B} \lambda_n (\Omega, B)& = \frac{1}{B} \left( \lambda_n(\Omega, B) +  \int_{ \Omega} B^2 \hat{A}^2 |u_n|^2  - |\nabla u_n|^2 \dd x \right), & \text{if } B \neq 0,\\
\frac{\partial}{\partial B} \lambda_n (\Omega, B)& = 0, & \text{if } B = 0,
\end{align} 
which follows from integration by parts. In the case of the constant magnetic field, we apply the above theorem to $\hat{A}(x)=1/2(-x_2,x_1,0)^T$.

The full gradient of the spectral objective function $J(c)$ is now obtained by chain rule and combination of the above formulas. If $\Omega_c$ is volume-normalized, i.e.\ $|\Omega_c| = 1$, then the partial derivatives of $J(c)$ are 
\begin{align} 
\begin{split}
\frac{\partial}{\partial c_i} J(c)&=\frac{\partial}{\partial c_i}\left(|\Omega_c|^\frac{2}{3} \lambda_1(\Omega_c, B|\Omega_c|^{-\frac{2}{3}})  \right)\\
&=\left(\frac{2}{3}|\Omega_c|^{-\frac{1}{3}} \lambda_1(\Omega_c, B|\Omega_c|^{-\frac{2}{3}}) - \frac{2}{3}B |\Omega_c|^{-\frac{5}{3}} \frac{\partial}{\partial B} \lambda_1(\Omega_c, B|\Omega_c|^{-\frac{2}{3}} ) \right) \frac{\partial |\Omega_c| }{ \partial c_i}  \\
& \qquad  - |\Omega_c|^\frac{2}{3} \int_{\partial \Omega_c}  \left| \frac{\partial u_1}{\partial n} \right|^2  V_i \cdot n \, d\sigma, \\
&=\int_{\partial \Omega_c} \left[  \frac{2}{3} \left(  \lambda_1(\Omega_c, B )- B \cdot \frac{\partial}{\partial B} \lambda_1(\Omega_c, B )\right)-  \left|\frac{\partial u_1}{\partial n} \right|^2 \right] V_i \cdot n \, d\sigma,
\end{split} \label{eq:shape_deriv_Jci}
\end{align}
where $\frac{\partial}{\partial B}\lambda_1(\Omega_c, B )$ is given by \eqref{eq:lambda_n_partialB} and $V_i$ must be a deformation field that deforms $\Omega$ according to changes of the $i$-th coefficient of the spherical harmonic expansion of $r(\theta,\varphi)$. It is not hard to see that such a deformation field must satisfy
\begin{align}
V_i(\gamma(\theta,\varphi)) = e_i^T \cdot \tilde{Y}(\theta,\varphi)  \begin{pmatrix} \sin \theta \cos \varphi \\  \sin \theta \sin \varphi \\  \cos \theta \end{pmatrix}.
\end{align}
We remark that the domain parametrization we chose guarantees that the domains have sufficient regularity for the derivative formulas to hold. \\

\textbf{Criticality condition and its consequence}\\

The shape derivative framework from above allows us to formulate a necessary condition for optimality of shapes with smooth boundary, namely that the first order shape derivative of the rescaled ground state energy must vanish. This leads to the following characterization of optimal sets.

\begin{proposition}[shape criticality condition] \label{prop:criticality_cond}
Suppose that $\Omega^*(B) \subset \mathbb{R}^3$ is a bounded, open minimizer of \eqref{eq:min_problem_c} with smooth boundary and assume that the ground state energy of $\Omega^*(B)$ is simple with corresponding ground state $u_1^*(B)$. Then, 
\begin{align*}
\left|\frac{\partial u_1^*(B)}{\partial n} \right|^2=  \frac{2}{3} \left(  \lambda_1^*( B )- B \cdot \frac{\partial}{\partial B} \lambda_1(\Omega^*(B), B )\right)  \qquad \text{ on } \partial \Omega^*(B) .
\end{align*}
\end{proposition}

By showing that the ball violates this optimality condition, see Lemma \ref{lem:u1_balls}, this allows to eliminate the ball as a minimizer of \eqref{eq:min_problem_c} for any $B>0$ (conditionally on the ground state simplicity assumption).

\begin{corollary} \label{cor:no_balls}
Let $B>0$ and assume that the ground state of the ball is simple. Then the ball is not a minimizer of \eqref{eq:min_problem_c}. 
\end{corollary}

\section{Results} \label{sec:res}

We generated numerical minimizers for the ground state energy $\lambda_1(\Omega,B)$ in two stages. In the first stage, we computed minimizers using the domain parametrization described in Section \ref{subsec:domainparam} and the full MPS approach described in Section \ref{subsec:mps}. This was done for $B\in [0.0, 50.0]$ in steps of $0.2$. In the second stage, for larger field strengths, we restricted ourselves to axisymmetric domains and switched to the axisymmetric MPS approach, employing the modifications mentioned in Section \ref{subsec:modaxis}. This was done for $B\in [44.0, 170.0]$ in steps of $0.2$ and was in our view necessary since the full MPS approach appeared to lose accuracy as the field strength increased and the axisymmetric MPS approach allows higher accuracy as explained in Section \ref{subsec:modaxis}. The overlap between the two stages was intentional to verify consistency of minimizers obtained with both approaches. We remark that for example at $B=50.0$, we have obtained $\lambda_1^*(B)\approx 55.455020$ from the full MPS approach and $\lambda_1^*(B)\approx 55.455053$ from the axisymmetric MPS approach, so both approaches produced minimizers with minimum eigenvalues that agreed to six significant digits.

Figures \ref{fig:lambda1_numA} and \ref{fig:lambda1_numB} compare the numerically obtained values for the minimal first eigenvalue $\lambda_1^*(B)$ with the first eigenvalue of a ball $\lambda_1(\mathbb{B},B)$ and the minimal first eigenvalue among cylinders $\lambda_{1,cyl}^*(B)$. One observes that the numerically obtained values for $\lambda_1^*(B)$ are close to $\lambda_1(\mathbb{B},B)$ for small $B$, but for large $B$, they are very close to $\lambda_{1,cyl}^*(B)$ and far from $\lambda_1(\mathbb{B},B)$. This is also illustrated by the following table:

\begin{center}
\begin{tabular}{r|rrr}
$B$ & $\lambda_1(\mathbb{B},B)$  & $\lambda_{1,cyl}^*(B)$  & $\lambda_1^*(B)$\\\hline
$10.0$ & $27.427$  & $29.658$ & $27.406$ \\
$20.0$  & $32.411$ & $34.133$ & $32.171$\\
$50.0$ & $57.736$ & $ 56.534$& $55.455$ \\
$100.0$ & $ 106.889$ & $102.894$ & $102.436$ \\
\end{tabular}
\end{center}

The difference between $\lambda_1(\mathbb{B},B)$ and $\lambda_1^*(B)$ seems to be steadily increasing with $B$. While at $B=10.0$, the difference between $\lambda_1(\mathbb{B},B)$ and $\lambda_1^*(B)$ is already numerically significant but still relatively small, at $B=100.0$, the difference between both values has increased to more than $4.4$. On the other hand, the difference between $\lambda_{1,cyl}^*(B)$ and $\lambda_1^*(B)$ is only about $0.4$ at $B=100.0$. The increasing difference between the minimal first eigenvalue and the first eigenvalue of a ball is owed to the fact that $\lambda_1(\mathbb{B},B)-B$ converges to an explicit, positive constant (see Appendix \ref{app:balls} for an asymptotic expansion of $\lambda_1(\mathbb{B},B)$), while $\lambda_{1,cyl}^*(B) - B$ and hence also $\lambda_{1}^*(B) - B$ tend to zero. This leads us to a first open problem we would like to pose.

\begin{openproblem}
If possible, find a two-term asymptotic expression for $\lambda_1^*(B)$ as $B\to\infty$. How does it compare to the asymptotic expression for $\lambda_{1,cyl}^*(B)$ from Theorem \ref{thm:lambda1cyl_asympt}?
\end{openproblem}

Figure \ref{fig:lambda1_numC} seems to suggest that the quotient
\begin{align}
\frac{ \lambda_{1}^*(B) - B}{\lambda_{1,cyl}^*(B)-B} 
\end{align}
does not go to zero as $B\to \infty$. This would imply that a possible second order term of $ \lambda_{1}^*(B)$ cannot be of lower order than the second order term of $\lambda_{1,cyl}^*(B)$. Since $\lambda_1^*(B) \leq \lambda_{1,cyl}^*(B)$ of course, we should perhaps expect that 
\begin{align}
\lambda_1^*(B) &= B + C \dfrac{\log(B)^2}{B^2} (1+o(1)) \qquad \text{as } B\to \infty,  
\end{align}
with $0<C\leq 36 \pi^4$.

Next, Figure \ref{fig:minimizers} shows the numerically obtained minimizers in the full range $B \in [0.0, 170.0]$ in steps of $10.0$. We observe that all minimizers appear to be axisymmetric. For $B>50.0$, this is of course unsurprising, since the symmetry is built-in. However, in the range $B \leq 50.0$, we have made no a priori symmetry assumption on the domains and we have included asymmetric initial domains for the minimization procedure. The consistent axisymmetry of the numerical minimizers is a result of the minimization procedure. Furthermore, we observe no discontinuity in the shape of the minimizers or in the plots of Figure \ref{fig:lambda1_num}. Hence, we think it is reasonable to assume that the axisymmetry found for smaller field strengths persists for larger field strengths. In other words, we propose the following conjecture.

\begin{conjecture} \label{conj:axis}
For any $B\geq 0$, the minimizer $\Omega^*(B)$ is axisymmetric with respect to the $z$-axis.
\end{conjecture}

Another emergent property from the minimization procedure is convexity. We have not built in or demanded convexity in either of the two optimization stages. Our numerical results lead us to believe

\begin{conjecture} \label{conj:convex}
For any $B\geq 0$, the minimizer $\Omega^*(B)$ is convex.
\end{conjecture}

Furthermore, Figure \ref{fig:minimizers} displays the ``spaghettification'' phenomenon discussed in Section \ref{sec:cyl} very well. The diameter of the numerical minimizers starts to increase as soon as the field strength $B$ becomes non-zero. In Section \ref{sec:cyl} we have estimated the speed with which the minimizers degenerate - partially under the assumption that the optimal domains stay convex, i.e.\ under assumption of the previous conjecture. Admittedly, these bounds were rather crude. We are left with another open problem.

\begin{openproblem}
Find stronger estimates on how the minimizers $\Omega^*(B)$ degenerate. If possible, find asymptotic expressions for $h(\Omega^*(B))$, $R(\Omega^*(B))$, $\operatorname{diam}(\Omega^*(B))$ and $r_{in}(\Omega^*(B))$ as $B\to \infty$.
\end{openproblem}

Figure \ref{fig:diam_opt} shows numerically obtained values of $h(\Omega^*(B))$ and $R(\Omega^*(B))$ compared to the height $h^*(B)$ and radius $R^*(B)=(\pi h^*(B))^{-1/2}$ of optimal cylinders from Section \ref{sec:cyl}. We observe that $h(\Omega^*(B)) \geq h^*(B)$ and $R(\Omega^*(B)) \geq R^*(B)$ for all field strengths $B$ considered. Figure \ref{fig:diam_opt} and the fact that the numerical minimizers in Figure \ref{fig:minimizers} look more and more like a long cylinder with round caps attached at both ends speaks in favour of the following conjecture.

\begin{conjecture}
$h(\Omega^*(B)) = h^*(B)(1+o(1))$ and $R(\Omega^*(B)) = R^*(B) (1+o(1))$ as $B\to \infty$.
\end{conjecture}

Finally, related to this observation, we think another natural question could be the following.

\begin{openproblem}
Do the minimizers $\Omega^*(B)$ converge as $B\to\infty$ after suitable rescaling? 
\end{openproblem}

We think that the answer to this problem should be positive and we suspect that appropriately rescaled minimizers converge to a cylinder.

\newpage

\begin{figure}

\begin{subfigure}{0.8\textwidth}
\begin{center}
\includegraphics[scale=0.4]{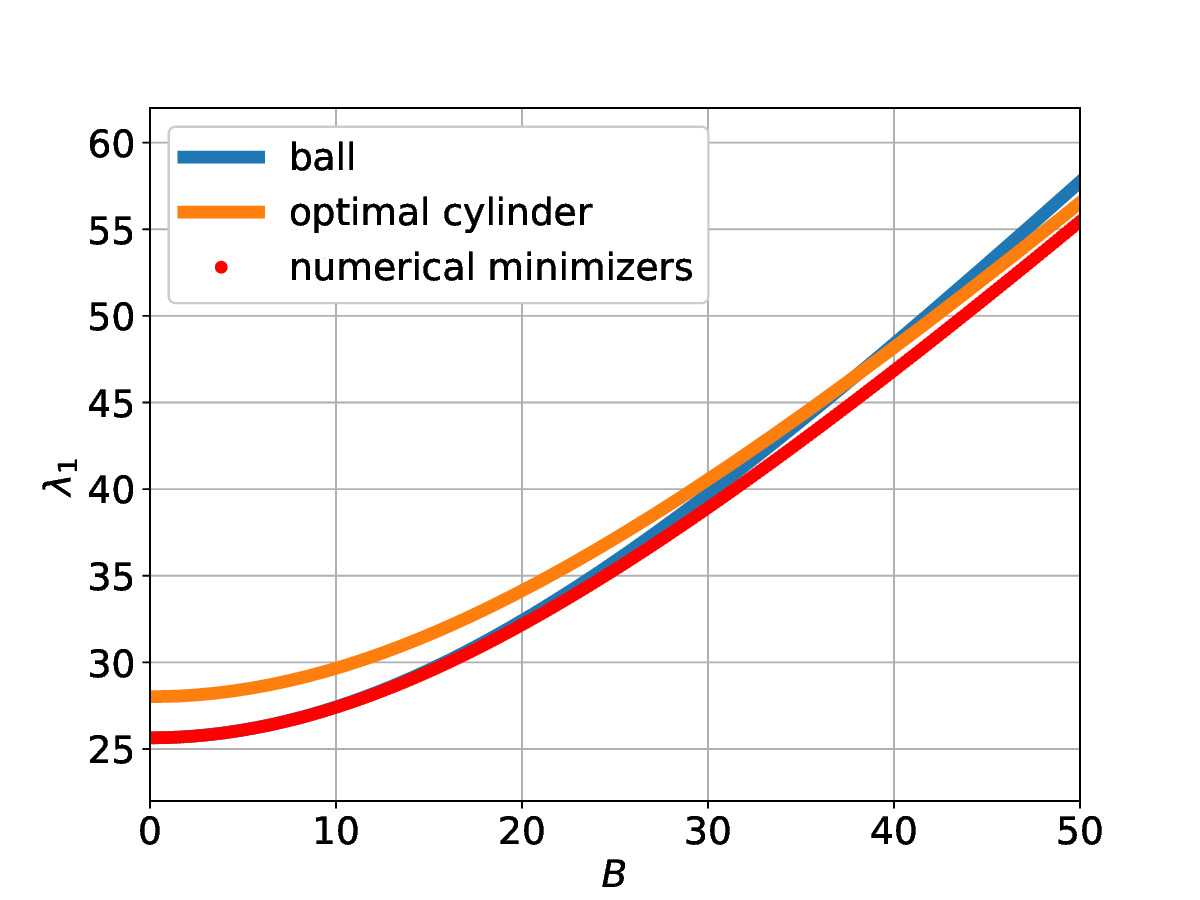}
\end{center}
\caption{The numerically obtained minimal first eigenvalue $\lambda_1^*(B)$ (red) for various field strengths in comparison to the first eigenvalue of a ball $\lambda_1(\mathbb{B}_R,B)$ (blue) and the minimal first eigenvalue among cylinders $\lambda_{1,cyl}^*(B)$ (orange).} \label{fig:lambda1_numA}
\end{subfigure}

\begin{subfigure}{0.8\textwidth}
\begin{center}
\includegraphics[scale=0.4]{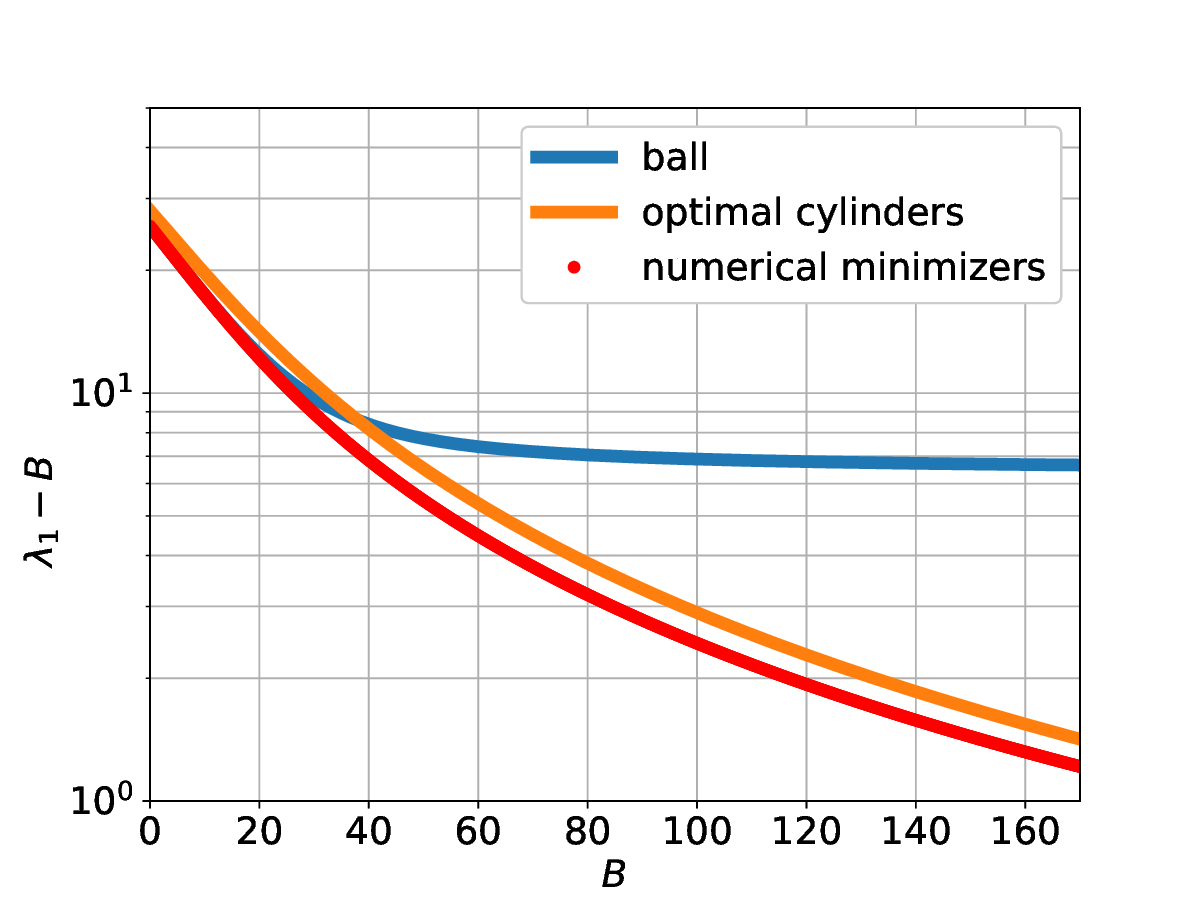}
\end{center}
\caption{Same as the plot (a) but shows the difference $\lambda_1- B$ on a logarithmically scaled axis.} \label{fig:lambda1_numB}
\end{subfigure}
\begin{subfigure}{0.8\textwidth}
\begin{center}
\includegraphics[scale=0.4]{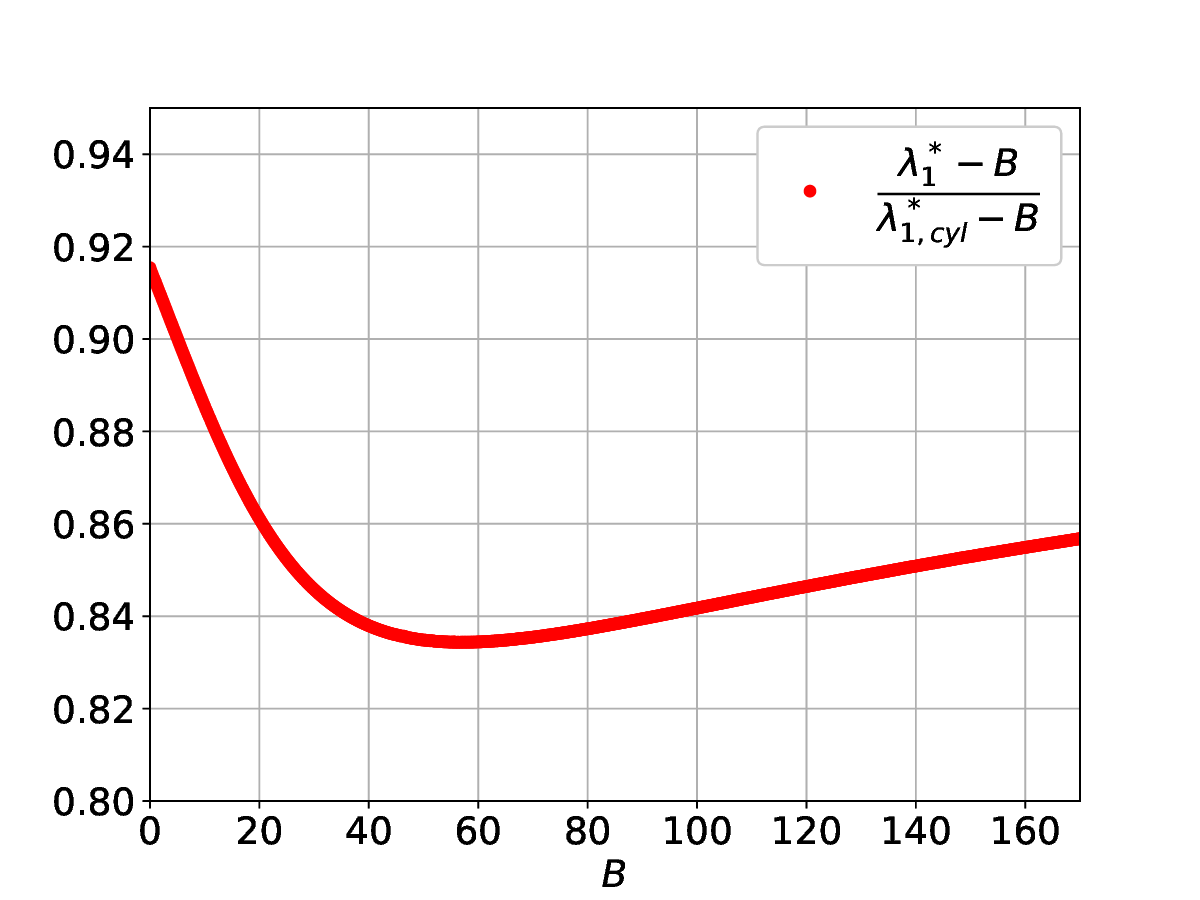}
\end{center}
\caption{Numerical values for the quotient $\dfrac{\lambda_{1}^*(B)- B}{\lambda_{1,cyl}^*(B)-B}$.} \label{fig:lambda1_numC}
\end{subfigure}

\caption{ } \label{fig:lambda1_num}

\end{figure}

\newpage

\begin{figure}

\begin{center}

\begin{overpic}[trim={1cm 1cm 1cm 1cm},clip,scale=0.14]{./figures/img\_lambda1B0.0.png}
 \put (30,-15) {$B=0.0$}
\end{overpic}
\begin{overpic}[trim={1cm 1cm 1cm 1cm},clip,scale=0.14]{./figures/img\_lambda1B10.0.png}
 \put (30,-15) {$B=10.0$}
\end{overpic}
\begin{overpic}[trim={1cm 1cm 1cm 1cm},clip,scale=0.14]{./figures/img\_lambda1B20.0.png}
 \put (30,-15) {$B=20.0$}
\end{overpic}
\begin{overpic}[trim={1cm 1cm 1cm 1cm},clip,scale=0.14]{./figures/img\_lambda1B30.0.png}
 \put (30,-15) {$B=30.0$}
\end{overpic}

\vspace{1cm}

\begin{overpic}[trim={1cm 1cm 1cm 1cm},clip,scale=0.14]{./figures/img\_lambda1B40.0.png}
 \put (30,-15) {$B=40.0$}
\end{overpic}
\begin{overpic}[trim={1cm 1cm 1cm 1cm},clip,scale=0.14]{./figures/img\_lambda1B50.0.png}
 \put (30,-15) {$B=50.0$}
\end{overpic}
\begin{overpic}[trim={1cm 1cm 1cm 1cm},clip,scale=0.14]{./figures/img\_lambda1B60.0.png}
 \put (30,-15) {$B=60.0$}
\end{overpic}
\begin{overpic}[trim={1cm 1cm 1cm 1cm},clip,scale=0.14]{./figures/img\_lambda1B70.0.png}
 \put (30,-15) {$B=70.0$}
\end{overpic}

\vspace{1cm}

\begin{overpic}[trim={1cm 1cm 1cm 1cm},clip,scale=0.14]{./figures/img\_lambda1B80.0.png}
 \put (30,-15) {$B=80.0$}
\end{overpic}
\begin{overpic}[trim={1cm 1cm 1cm 1cm},clip,scale=0.14]{./figures/img\_lambda1B90.0.png}
 \put (30,-15) {$B=90.0$}
\end{overpic}
\begin{overpic}[trim={1cm 1cm 1cm 1cm},clip,scale=0.14]{./figures/img\_lambda1B100.0.png}
 \put (30,-15) {$B=100.0$}
\end{overpic}
\begin{overpic}[trim={1cm 1cm 1cm 1cm},clip,scale=0.14]{./figures/img\_lambda1B110.0.png}
 \put (30,-15) {$B=110.0$}
\end{overpic}

\vspace{1cm}

\begin{overpic}[trim={1cm 1cm 1cm 1cm},clip,scale=0.14]{./figures/img\_lambda1B120.0.png}
 \put (30,-15) {$B=120.0$}
\end{overpic}
\begin{overpic}[trim={1cm 1cm 1cm 1cm},clip,scale=0.14]{./figures/img\_lambda1B130.0.png}
 \put (30,-15) {$B=130.0$}
\end{overpic}
\begin{overpic}[trim={1cm 1cm 1cm 1cm},clip,scale=0.14]{./figures/img\_lambda1B140.0.png}
 \put (30,-15) {$B=140.0$}
\end{overpic}
\begin{overpic}[trim={1cm 1cm 1cm 1cm},clip,scale=0.14]{./figures/img\_lambda1B150.0.png}
 \put (30,-15) {$B=150.0$}
\end{overpic}

\vspace{1cm}

\begin{overpic}[trim={1cm 1cm 1cm 1cm},clip,scale=0.14]{./figures/img\_lambda1B160.0.png}
 \put (30,-15) {$B=160.0$}
\end{overpic}
\begin{overpic}[trim={1cm 1cm 1cm 1cm},clip,scale=0.14]{./figures/img\_lambda1B170.0.png}
 \put (30,-15) {$B=170.0$}
\end{overpic}

\vspace{0.7cm}

\end{center}

\caption{Minimizers $\Omega^*(B)$ for various field strengths. The magnetic field is oriented along the vertical axis (drawn in red).}\label{fig:minimizers}

\end{figure}

\clearpage

\begin{figure}
\includegraphics[scale=0.33]{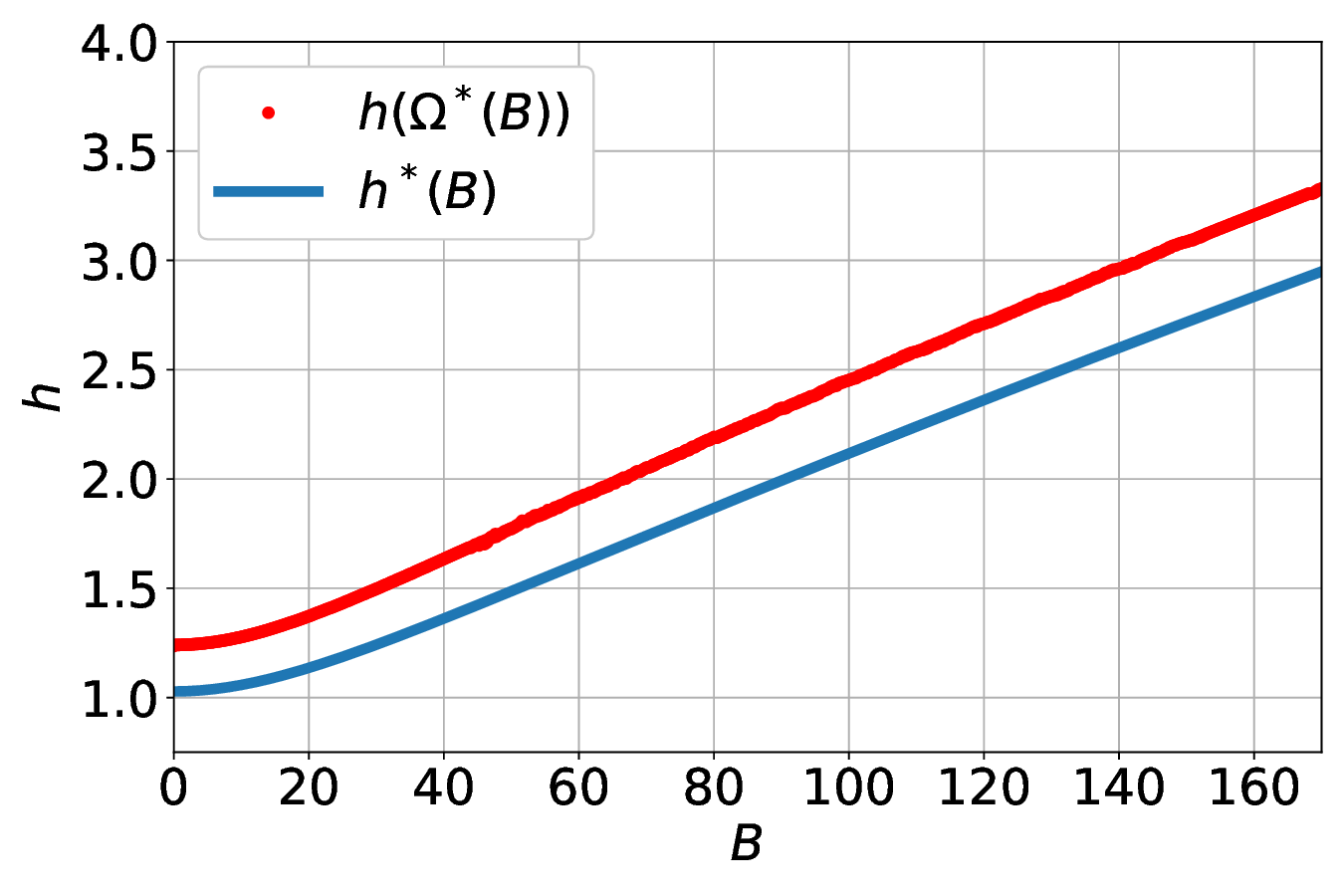}
\includegraphics[scale=0.33]{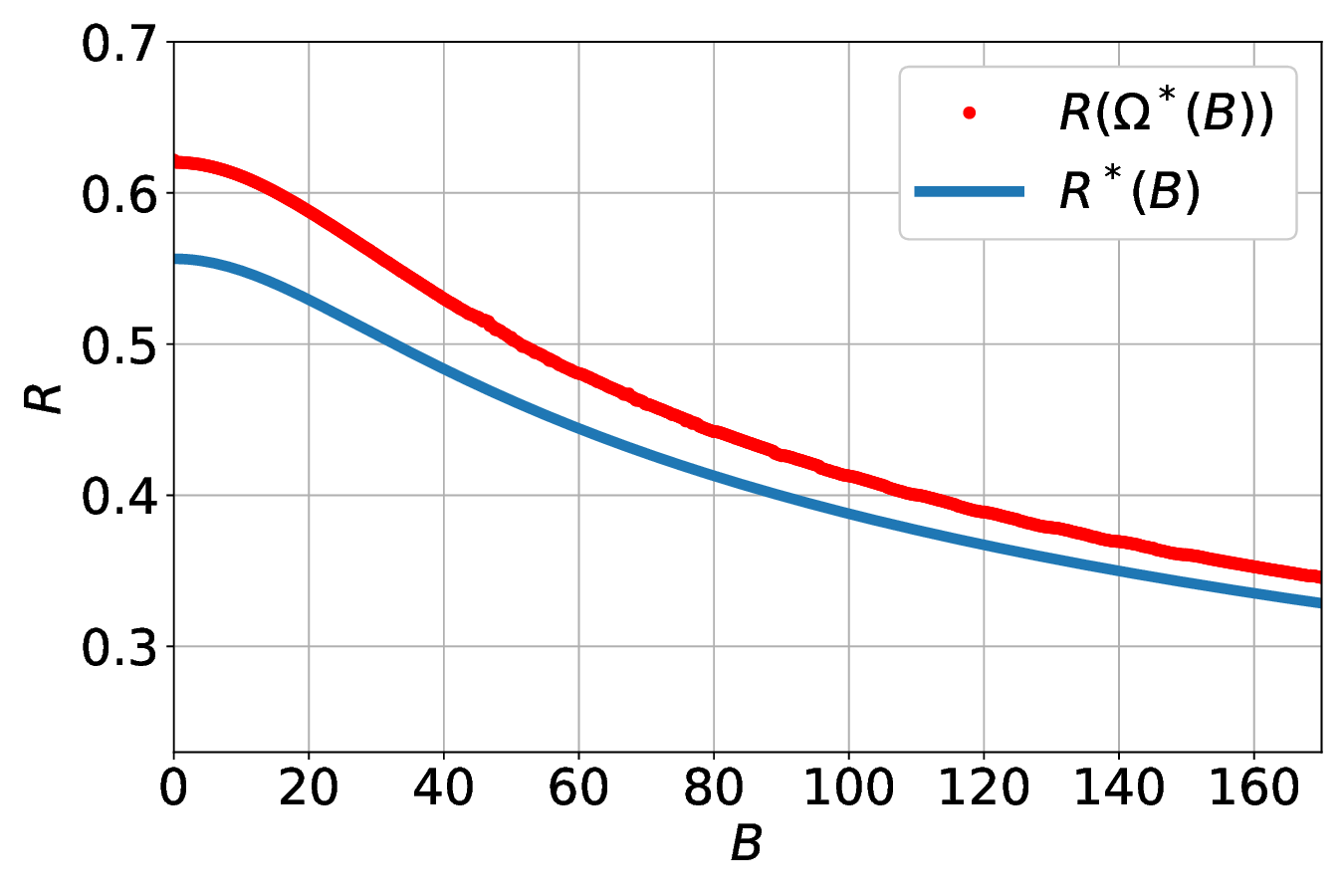}
\caption{Numerically obtained values for the lengths $h(\Omega^*(B))$ and $R(\Omega^*(B))$ compared to the lengths $h^*(B)$ and $R^*(B)=(\pi h^*(B))^{-1/2}$ for optimal cylinders.} \label{fig:diam_opt}
\end{figure}

\appendix

\section{Existence of minimizers} \label{app:existence_of_minimizers}

Focusing on minimizing just the lowest eigenvalue of the magnetic Dirichlet Laplacian, we provide in the following a simplified existence argument for a minimizing shape. First, we generalize the notion of the ground state energy to quasi-open sets. For a bounded, quasi-open set $\Omega\subset\mathbb R^3$ of finite measure, we define
\[
 \lambda_1(\Omega, B)
 :=
 \inf_{\substack{
        u\in H^1(\mathbb R^3;\mathbb C),\\
        \|u\|_{L^2(\mathbb R^3)}=1 \\
        u=0\ \mathrm{q.e.~on}\ \mathbb{R}^3 \setminus \Omega }}
 \int_{\mathbb R^3} |(-i\nabla+A)u|^2\,dx .
\]
Here, q.e.~means quasi-everywhere, and $u$ is understood through its quasi-continuous representative (see \cite{Buttazzo1993, Buttazzo2011} for definitions of quasi-open, -everywhere and -continuous; every $u \in H^1(\mathbb{R}^3; \mathbb{C})$ has a quasi-continuous representative that is quasi-everywhere unique, see \cite[Theorem 3.3.29]{Henrot2018}). For open domains, this generalized definition of the ground state energy coincides with the usual form domain formulation. Then, we have the following existence result.

\begin{proposition} \label{prop:exist}
Let $D\subset \mathbb{R}^3$ be bounded and open. Assume \(0< c \leq |D|\) and  \(B\geq 0\). Then the relaxed problem
\begin{align}
 \inf\Big\{   \lambda_1(\Omega, B): \Omega\subset D \ \text{quasi-open},\ |\Omega| \leq c \Big\} \label{eq:prop_exist}
\end{align}
admits a minimizer among quasi-open sets.
\end{proposition}

It is important to emphasize that this only guarantees existence in the relaxed class of quasi-open sets contained in a bounded design region $D$. It guarantees neither that the minimizer is unique nor that it is open, smooth, connected or convex.

\begin{proof}[Proof of Proposition \ref{prop:exist}]
Let $B$ be fixed and let $ \lambda_1^*$ denote the infimum in \eqref{eq:prop_exist}. In the following, we use the equivalent functional formulation
\[
 \lambda_1^*
 =
 \inf\left\{
        \int_{\mathbb R^3} |(-i\nabla+A)u|^2\,dx: u\in H^1(\mathbb R^3;\mathbb C), \ 
        \|u\|_{L^2(\mathbb R^3)}=1,\ \{\widetilde{|u|}>0\} \subset D, \ 
        |\{\widetilde{|u|}>0\}|\leq c
      \right\},
\]
where \(\widetilde{|u|}\) denotes the quasi-continuous representative of
\(|u| \). Indeed, if \(u\) is admissible for a quasi-open set \(\Omega\), then
\(\{\widetilde{|u|}>0\}\subset\Omega\) up to a set of capacity zero.  Conversely,
\(\{\widetilde{|u|}>0\}\) is quasi-open (see \cite[Proposition 3.3.41]{Henrot2018}) and \(u\) is admissible for this set. Also, since the vector potential $A$ is bounded in $D$, the norms of the magnetic Sobolev space $H_A^1(D)$ and the usual Sobolev space $H^1(D)$ are equivalent. Therefore, we may also write
\[
 \lambda_1^*
 =
 \inf\left\{
        \int_{\mathbb R^3} |(-i\nabla+A)u|^2\,dx: u\in H^1_0(D), \ 
        \|u\|_{L^2(\mathbb R^3)}=1, \ 
        |\{\widetilde{|u|}>0\}|\leq c
      \right\},
\]
where we see $H_0^1(D) = \{ u \in H^1(\mathbb{R}^3; \mathbb{C}) \, : \, u=0\ \text{q.e.~on}\ \mathbb{R}^3 \setminus D \}$ as a subspace of $H^1(\mathbb{R}^3; \mathbb{C})$, cf.~\cite[Theorem 3.3.42]{Henrot2018}.

Now let \((u_n)_{n\in\mathbb{N}}\) be a minimizing sequence for the functional formulation, satisfying
\[
     u_n \in H_0^1(D), \qquad   \|u_n\|_{L^2(\mathbb R^3)}=1,
        \qquad
        |\{\widetilde{|u_n|}>0\}|\leq c,
\]
and
\[
        \int_{\mathbb R^3} |(-i\nabla+A)u_n|^2\,dx
        \to \lambda_1^* \qquad ( n \to \infty).
\]
Since, the norms of the magnetic Sobolev space $H_A^1(D)$ and the usual Sobolev space $H^1(D)$ are equivalent, the sequence  $(u_n)_{n\in\mathbb{N}}$ is bounded in $H^1(\mathbb{R}^3;\mathbb{C})$. Moreover, $H_0^1(D)$ is a closed linear subspace of $H^1(\mathbb{R}^3;\mathbb{C})$, hence, we can go over to a subsequence of $(u_n)_{n\in\mathbb{N}}$ such that $ u_n \rightharpoonup u $ weakly in $H^1(\mathbb{R}^3;\mathbb{C})$ and $u \in H_0^1(D)$. Furthermore, $H_0^1(D)$ is compactly embedded into $L^2(D)$, so passing to another subsequence, we can assume that also $u_n|_D \to u|_D$ strongly in $L^2(D)$, thus $\Vert u \Vert_{L^2(D)} =1 $. However, because $u_n=0$ and $u=0 $ q.e.~on $\mathbb{R}^3 \setminus D$, we also have $u_n \to u $ strongly in $L^2(\mathbb{R}^3)$ and $\Vert u \Vert_{L^2(\mathbb R^3)} = 1$. After passing to a further subsequence, we may assume that $u_n\to u$ almost everywhere in $\mathbb R^3$. Since the quasi-continuous representatives agree almost everywhere with the corresponding Sobolev functions, we have
\[
 \mathbf 1_{\{\widetilde{|u|}>0\}}
 \leq
 \liminf_{n\to\infty}
 \mathbf 1_{\{\widetilde{|u_n|}>0\}}
\]
almost everywhere in $\mathbb R^3$. Therefore, by Fatou's lemma,
\[
 |\{\widetilde{|u|}>0\}|
 \leq
 \liminf_{n\to\infty}
 |\{\widetilde{|u_n|}>0\}|
 \leq c.
\]
By weak lower semicontinuity of the magnetic quadratic form on bounded sets, in particular on $D$, we obtain
\[
        \int_{\mathbb R^3}|(-i\nabla+A)u|^2\,dx
        \leq
        \liminf_{n\to\infty}
        \int_{\mathbb R^3}|(-i\nabla+A)u_n|^2\,dx
        =
        \lambda_1^*.
\]
Thus \(u\) realizes the minimum in the functional formulation. Now define
\[
        \Omega^* := D \cap \{\widetilde{|u|}>0\}.
\]
As $\widetilde{|u|}$ is quasi-continuous, the set $\Omega^*$ is quasi-open (see \cite[Proposition 3.3.41]{Henrot2018}) and by construction $ \Omega^* \subset D$ and \(|\Omega^*| \leq |\{\widetilde{|u|}>0\}|\leq c\). Since $u$ is admissible for $\Omega^*$, we have
\[
        \lambda_1(\Omega^*,B)
        \leq
        \int_{\mathbb R^3}|(-i\nabla+A)u|^2\,dx
        =
        \lambda_1^*.
\]
Hence, $\lambda_1(\Omega^*,B)=\lambda_1^*$ and \(\Omega^*\) is a minimizer of \eqref{eq:prop_exist}.
\end{proof}

For the exact-volume problem, existence of minimizers follows immediately: If one finds a minimizer of \eqref{eq:prop_exist} having volume equal to $c$, there is nothing to prove. If one finds a minimizer of \eqref{eq:prop_exist} with volume strictly less than $c$, one can choose a larger quasi-open set in $D$ which contains the minimizer and some additional region which accounts for the missing volume. By domain monotonicity (see Proposition \ref{prop:simple_est_domain_mono} (4)), the ground state energy of the enlarged domain cannot increase, but strict decrease is forbidden, since the original domain was already a minimizer. Therefore, the enlarged domain is also a minimizer and the relaxed exact-volume problem also admits a minimizer.

Finally, we remark that if a quasi-open minimizer $\Omega^*$ of \eqref{eq:prop_exist} turns out to be convex, then there also exists an open minimizer  of \eqref{eq:prop_exist}. This is a consequence of the following lemma. 

\begin{lemma} \label{lem:convex_cap}
If $\Omega \subset \mathbb{R}^d$ is bounded, quasi-open and convex, then $\operatorname{cap}(\Omega \setminus \operatorname{int} \Omega) = 0$.
\end{lemma}

If $|\Omega^*|>0$ and therefore $\operatorname{int} \Omega^* \neq  \varnothing$, the lemma implies that $\operatorname{int} \Omega^*$ agrees with $\Omega^*$ up to a set of zero capacity and Lebesgue measure. The interior $\operatorname{int} \Omega^*$ is thus an open representative with the same Dirichlet eigenvalues as $\Omega^*$. Any quasi-open, convex minimizer can therefore be replaced by its open representative. For completeness, we give the proof of the lemma.

\begin{proof}[Proof of Lemma \ref{lem:convex_cap}]
Since $\Omega$ is convex, it holds $|\Omega \setminus (\operatorname{int} \Omega)| = 0$ by Lebesgue's density theorem. Moreover, since $\Omega$ is quasi-open, it can be identified up to a set of capacity zero with the superlevel set $\{ \tilde{u}>0\}$ of a nonnegative function $u\in H^1(\mathbb{R}^d)$, see \cite{Russ2019, Buttazzo2011}. This function hence satisfies $\tilde{u}=0$ q.e.~on $\mathbb{R}^d \setminus \Omega$, thus $\tilde{u}=0$ a.e.~on $\mathbb{R}^d \setminus \Omega$ and because $|\Omega \setminus (\operatorname{int} \Omega)| = 0$, we have $\tilde{u}=0$ a.e.~on $\mathbb{R}^d \setminus ( \operatorname{int} \Omega)$. We now distinguish two cases. If $|\Omega|>0$, then $\Omega$ has a Lipschitz boundary, since it is convex, and we conclude that $u \in H_0^1(\operatorname{int} \Omega)$, see Proposition 3.2.16 of \cite{Henrot2018}. If $|\Omega| = 0$, then $\operatorname{int} \Omega = \varnothing$ and $u=0$ a.e.~on $\mathbb{R}^d$. But in both cases, as $\operatorname{int} \Omega$ is open, Theorem 3.3.42 of \cite{Henrot2018} implies that $\tilde{u}=0$ q.e.~on $\mathbb{R}^d \setminus ( \operatorname{int} \Omega)$ and in particular, $\operatorname{cap}(\{ \tilde{u}>0\}\setminus ( \operatorname{int} \Omega) ) = 0$. However, $\Omega$ coincides with $\{ \tilde{u}>0\}$ up to a set of capacity zero and capacity is subadditive, therefore $\operatorname{cap}(\Omega\setminus ( \operatorname{int} \Omega) ) = 0$.
\end{proof}

%

\section{Balls} \label{app:balls}

In this section, we prove two auxiliary results for the ground state of the ball.

Let $\mathbb{B}_R :=\{ x\in\mathbb{R}^3 \,:\, |x|<R\}$ denote the ball in $\mathbb{R}^3$ of radius $R>0$ centered at the origin. We have the following lower bound and asymptotic expansion for the first eigenvalue of the magnetic Dirichlet Laplacian on $\mathbb{B}_R$.

\begin{proposition} \label{prop:app_ball}
\begin{align}
\lambda_1(\mathbb{B}_R, B)  \geq B + \left(\frac{\pi}{2R} \right)^2 \qquad \text{for any } B \geq 0
\end{align}
and 
\begin{align}
\lambda_1(\mathbb{B}_R, B) = B + \left(\frac{\pi}{2R} \right)^2 + O\left( \frac{\log(B)}{B} \right) \qquad \text{as } B\to \infty.
\end{align}
\end{proposition}

\begin{proof} 
The lower bound on $\lambda_1(\mathbb{B}_R, B)$ follows directly from Lemma \ref{lem:basic_est_pluspih2}.

For the upper bound, we consider a circular cylinder of radius $0<R(B)< R$ and height $h(B)=2\sqrt{R^2-R(B)^2}$. After a suitable translation, such a cylinder can be contained within $\mathbb{B}_R$ and thus
\begin{align}
\lambda_1(\mathbb{B}_R, B) \leq \lambda_1(C_{R(B),h(B)}, B) = \lambda_1(D_{R(B)},B) + \frac{\pi}{h(B)^2}. \label{eq:ball_upperbnd}
\end{align}
We now choose 
\begin{align}
R(B) = \sqrt{4 \frac{\log(B)}{B}}, \qquad B\geq e.
\end{align}
Then, 
\begin{align}
\frac{\pi^2}{h(B)^2} = \left(\frac{\pi}{2R} \right)^2 \left( \frac{R^2}{R^2-R(B)^2}  \right) = \left(\frac{\pi}{2R} \right)^2 +O(R(B)^2) = \left(\frac{\pi}{2R} \right)^2 + O\left( \frac{\log(B)}{B}  \right) \label{eq:ball_upperbnd_1}
\end{align}
as $B\to\infty$ and by \cite[Proposition 2.5]{Ekholm2016},
\begin{align}
\lambda_1(D_{R(B)},B) &\leq B + e B^2 R(B)^2 \exp \left( - \frac{BR(B)^2}{2}\right) =B + 4e \frac{ \log(B) }{B}. \label{eq:ball_upperbnd_2}
\end{align}
Combining the estimates \eqref{eq:ball_upperbnd_1} and \eqref{eq:ball_upperbnd_2} in \eqref{eq:ball_upperbnd} together with the lower bound gives the desired result.
\end{proof}

\begin{remark}
If $|\mathbb{B}_R| =1$, then $R=(4\pi/3)^{-1/3}$ and hence  
\begin{align}
\lambda_1(\mathbb{B}_R, B) = B + \frac{\pi^{8/3}}{6^{2/3}}+  O\left( \frac{\log(B)}{B} \right), 
\end{align}
as $B\to\infty$. Here, $\pi^{8/3}/6^{2/3} \approx 6.412$.
\end{remark}

The second result allows excluding the ball as a minimizer of \eqref{eq:min_problem_c} for any $B>0$ by means of Proposition \ref{prop:criticality_cond}.

\begin{lemma} \label{lem:u1_balls}
Let $B>0$. Suppose that $\lambda_1(\mathbb{B}_R, B)$ is simple and let $u_1$ denote the corresponding ground state. Then $\left|\frac{\partial u_1}{\partial n} \right|$ is not constant on $\partial \mathbb{B}_R$.
\end{lemma}

\begin{proof}
Assume for the sake of contradiction that $\left|\frac{\partial u_1}{\partial n} \right|$ is constant on $\partial \mathbb{B}_R$ and let us denote this constant by $C$. First, we note that if $C=0$, then we would have $u_1 = 0 $ as well as $\frac{\partial u_1}{\partial n} = 0$ on $\partial \mathbb{B}_R$, and thus $u_1 \equiv 0$ by \cite[Lemma 2.1]{Lotoreichik2025a}, which contradicts that $u_1$ is a non-trivial eigenfunction. So let us assume $C>0$ in the following. Simplicity of the ground state implies that $u_1$ is axisymmetric and therefore
\begin{align*}
(-i\nabla +A)^2 u_1 = \left( - \Delta + \frac{B^2}{4} (x^2+y^2) \right) u_1 .
\end{align*}
Hence, $u_1$ can be identified with the ground state of a Schrödinger operator with anisotropic harmonic potential and we may choose $u_1$ as a positive function on $\mathbb{B}_R$. 

Let us now define $u_1' $ by swapping the $x$ and $z$ coordinates, i.e.
\begin{align*}
u_1'(x,y,z) = u_1(z,y,x),
\end{align*}
and let 
\begin{align*}
L=- \Delta + \frac{B^2 }{4} (x^2+y^2) - \lambda_1(\mathbb{B}_R,B).
\end{align*}
Then, $L\geq 0$ as an operator on $\mathbb{B}_R$, $Lu_1 = 0$ and $Lu_1' = \frac{B^2}{4} (x^2-z^2) u_1'$. Let us furthermore introduce $v:= u_1'- u_1$ and the set $S:= \mathbb{B}_R \cap \{ x>|z| \}$. Then, $Lv = \frac{B^2}{4} (x^2-z^2) u_1' > 0$ on $S$. We further claim that $v=0$ on $\partial S$. Indeed, $v=0$ on $\partial \mathbb{B}_R$ since $u_1$, $u_1' = 0$ on $\partial \mathbb{B}_R$ and $v=0$ on $\{x=z\}$ since $u_1 = u_1'$ on $\{x=z\}$ by construction. On $\{x=-z \}$, we have $u_1'(x,y,z)=u_1'(x,y,-x)=u_1(-x,y,x)=u_1(x,y,x)$ by definition and axisymmetry. Moreover, since $L$ is invariant under reflection with respect to the $xy$-plane and $u_1$ is simple, the ground state $u_1$ must be symmetric with respect to reflections across the $xy$-plane. Therefore, $u_1'(x,y,z)=u_1(x,y,x)=u_1(x,y,-x)=u_1(x,y,z)$ on $\{x=-z \}$ and thus $v=0$ on $\{x=-z \}$.

Now, $S$ is a proper subset of $\mathbb{B}_R$ and by the variational principle $\lambda_1(S,B) > \lambda_1(\mathbb{B}_R, B)$. Thus, $L>0$ as an operator on $S$ and $L$ satisfies a maximum principle. Then $Lv>0$ on $S$ and $v=0$ on $\partial S$ imply that $v>0$ in $S$. By Hopf's lemma, we must have $\frac{\partial v}{\partial n} < 0 $ on any smooth part of $\partial S$. 

However, we also have $\frac{\partial v}{\partial n} = \frac{\partial u_1'}{\partial n} -\frac{\partial u_1}{\partial n}$ on $\partial \mathbb{B}_R$ where $\frac{\partial u_1}{\partial n} = -C$ on $\partial \mathbb{B}_R$ by Hopf's lemma. But $u_1'$ is just a reflection of $u_1$, hence also $\frac{\partial u_1'}{\partial n} = -C $ on $\partial \mathbb{B}_R$, and so $\frac{\partial v}{\partial n} = 0$ on $\partial \mathbb{B}_R$. With this, we have derived a contradiction on the spherical part of $\partial S$ and we conclude that $\left|\frac{\partial u_1}{\partial n} \right|$ cannot be constant on $\partial \mathbb{B}_R$.
\end{proof}

\section*{Acknowledgements}

The author is grateful to Timo Weidl for proposing this problem to him and valuable discussions. He would like to thank James B. Kennedy for pointing out the proof strategy of the existence proof in Appendix \ref{app:existence_of_minimizers}. The author also thanks anonymous referees for helpful comments.


\section*{Declarations}

\begin{itemize}
\item Funding: No funding was received to assist with the preparation of this manuscript. 
\item Competing interests: The author has no competing interests to declare.
\item Ethics approval and consent to participate: Not applicable
\item Consent for publication: Not applicable
\end{itemize}

\section*{AI usage disclosure}

The author acknowledges the use of ChatGPT for the development of the proofs of Proposition \ref{prop:exist} and Lemma \ref{lem:u1_balls}. AI-assisted arguments were checked by the author, who takes full responsibility for the contents of the paper.

\section*{Code and Data}

The shape optimization code and minimizer data used to generate the plots and figures in the results section is available at \url{https://github.com/matthias-baur/mssopython3D}.

\printbibliography

\end{document}